\documentclass[conference]{IEEEtran}
\IEEEoverridecommandlockouts

\usepackage{cite}
\usepackage{subfigure}
\usepackage{amssymb,amsfonts}
\usepackage[centertags]{amsmath} 
\usepackage{algorithmic}
\usepackage{graphicx}
\usepackage{textcomp}
\usepackage{xcolor}
\usepackage{amsthm}
\usepackage{booktabs}
\usepackage{tabularx}
\usepackage{makecell}
\usepackage{multirow}
\theoremstyle{plain}

\newtheorem{proposition}{\textbf{Proposition}}
\newtheorem{remark}{\textbf{Remark}}
\newtheorem{corollary}{\textbf{Corollary}}
\def\BibTeX{{\rm B\kern-.05em{\sc i\kern-.025em b}\kern-.08em
    T\kern-.1667em\lower.7ex\hbox{E}\kern-.125emX}}
\begin{document}
%
\title{JSCGC: Joint Source-Channel-Generation Coding for Wireless Generative Communications}
\author{Tong Wu, Zhiyong Chen, \emph{Senior Member, IEEE}, Guo Lu, \emph{Member, IEEE}, Li Song, \emph{Senior Member, IEEE},\\ Feng Yang, Meixia Tao, \emph{Fellow, IEEE}, and Wenjun Zhang, \emph{Fellow, IEEE}
\thanks{Part of this work is to appear at IEEE ISIT-W 2026 \cite{isit-jscgc}.}
\thanks{The authors are with the Cooperative Medianet Innovation Center, the School of Information Science and Electronic Engineering, Shanghai Jiao Tong University, Shanghai 200240, China (Email: \{wu\_tong, zhiyongchen, luguo2014, song\_li, yangfeng, mxtao, zhangwenjun\}@sjtu.edu.cn). (Corresponding author: Zhiyong Chen.)}}
\maketitle

\begin{abstract}
Conventional communication systems, including both separation-based coding and learning-based joint source-channel coding (JSCC), are typically designed under Shannon’s rate-distortion theory. However, relying on generic distortion metrics fails to capture complex human visual perception, often resulting in blurred or unrealistic reconstructions. In this paper, we propose Joint Source-Channel-Generation Coding (JSCGC), a generative communication paradigm that replaces the conventional decoder with a generative model at the receiver. The received signal is treated as a condition that controls the sampling process into the learned conditional distribution, reformulating communication from deterministic reconstruction for distortion minimization to controlled generation for mutual information maximization under perceptual constraints.  Based on this formulation, we develop a unified joint training and efficient stochastic sampling framework, and provide theoretical analysis of its effectiveness in both learning and inference stages. Extensive experiments on latent-space image transmission demonstrate that the JSCGC consistently improves feature-based, semantic-level, and distributional quality across diverse channel conditions, while exhibiting a distinct error behavior characterized by semantic inconsistency rather than distortion.
\end{abstract}

\section{Introduction}
For decades, wireless communication system design has been largely shaped by the reconstruction paradigm rooted in Shannon's rate-distortion (RD) theory \cite{RD}. In this framework, the receiver is designed to recover a point estimate of the source signal by minimizing a predefined distortion metric under a given bit rate constraint, as shown in Fig. 1(a). Following this principle, classical communication approaches and recent deep joint source-channel coding (JSCC) methods \cite{gundu2019, NTSCC, MambaJSCCWu, video} have primarily focused on optimizing explicit distortion measures like mean squared error (MSE) and perceptual alignment metrics.

However, optimizing toward a specific distortion metric may introduce metric-induced bias. As stated by Goodhart's law, when a proxy measure becomes the optimization target, it may no longer faithfully reflect the underlying objective \cite{goodhart}. In the wireless communication system, this can lead to schemes that achieve lower distortion values but fail to preserve the visual quality and relevant information \cite{LPIPSZhang}. For example, optimization for MSE often produces overly smooth reconstructions, while optimization for learned perceptual image patch similarity (LPIPS) may introduce structural artifacts \cite{HiFiC,Krawczy}. These issues become more pronounced under low-rate or highly noisy channel conditions.


To address these limitations, rate-distortion-perception (RDP) based approaches \cite{RDPJun1, RDPJun2, RDPC, PDJSCC} have introduced perceptual constraints into the training objective, typically through adversarial losses, to better balance distortion fidelity and perceptual quality. More recently, generative models have been incorporated into the receiver to improve decoding quality \cite{CDDM,ICDM,SGD-JSCC,DiffCom,DeepJSCC-Diff,CommIN,DiffJSCC}. Representative methods such as DiffCom \cite{DiffCom}, DeepJSCC-Diff \cite{DeepJSCC-Diff}, and CommIN \cite{CommIN} utilize features from JSCC to condition pre-trained generative models for generating images, while DiffJSCC \cite{DiffJSCC} further adapts diffusion models to communication scenarios through fine-tuning. Although these methods significantly improve perceptual quality, they still rely on distortion-shaped representations, and the generation process remains conditioned on reconstructed signals, thus not fully removing the influence of distortion-based design.

This limitation reflects a more fundamental issue in reconstruction-based communication systems, where performance is evaluated through proxy distortion objectives. These objectives inevitably shape the learned point-estimation-based encoding and decoding strategies, leading systems to optimize surrogate metrics rather than the content quality.
\begin{figure}[t]
\centerline{\includegraphics[width=0.45\textwidth]{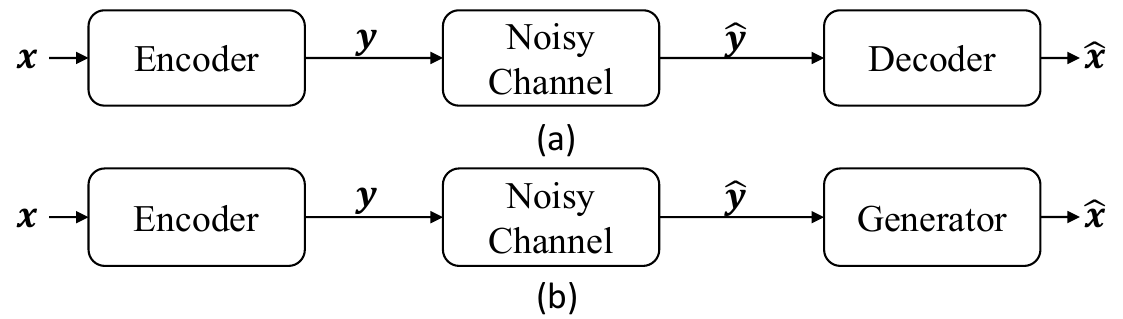}}
\caption{System model: (a) reconstruction paradigm; (b) generation paradigm.}
\label{System_model}
\end{figure}

Recent advances in generative foundation models \cite{lipman2023flow} have opened new possibilities for communication system design. Diffusion models and multimodal generative models have demonstrated an unprecedented ability to learn the distribution of natural data through large-scale pretraining, thereby acquiring powerful content generation capabilities. Under such a paradigm, communication systems may no longer need to transmit all information required for reconstruction. Instead, it becomes sufficient to convey semantic conditions that guide the generative process at the receiver, allowing the powerful generative model to synthesize high-quality content consistent with the source. This emerging capability suggests that communication systems can be redesigned to convey semantic conditions for generation rather than information for reconstruction, providing a new opportunity for generation-oriented communication paradigms.

Motivated by the observation, we propose a new generative communication paradigm termed \textbf{Joint Source-Channel-Generation Coding (JSCGC)}. As shown in Fig. 1(b), the core idea is to replace the conventional decoder with a generator, thereby shifting the paradigm from reconstruction to controlled generation. In this paradigm, the received signal is no longer used for direct reconstruction, but treated as a semantic condition that guides sampling from the conditional distribution. Correspondingly, the transmitter is no longer designed for reconstruction but for controlling the generative process. Therefore, communication is reformulated from distortion minimization to controlled generation guided by the received signal.

Within this paradigm, we first formulate the communication objective as mutual information (MI) maximization under a perceptual constraint, without relying on explicit distortion functions. Building on this formulation, we derive a tractable optimization objective via variational inference and Lagrangian relaxation, enabling end-to-end learning of the encoder, transmission, and generator. We further develop an efficient conditional sampling strategy to enable high-quality generation under limited communication resources. The resulting scheme is theoretically proven to be consistent with the underlying information-theoretic formulation, as it simultaneously promotes mutual information maximization and perceptual consistency in both training and inference stages. To implement the proposed framework, we design an end-to-end system that combines a Mamba-based encoder \cite{MambaJSCCWu} with a latent flow matching model \cite{z-image}. The two components are connected via a communication-aware adapter, enabling controllable generation from received signals while preserving the prior knowledge learned by large-scale generative models.


The main contributions of this paper are summarized as follows:
\begin{itemize}
\item \textbf{JSCGC Framework and Implementation}. We propose JSCGC, a joint source-channel-generation communication framework that replaces deterministic reconstruction with controlled generative communication. The communication objective is formulated as mutual information maximization under a perceptual constraint, avoiding explicit distortion functions during training. We further present an implementation of JSCGC for latent-space image transmission using a Mamba-based encoder, a latent flow matching model, and a communication-aware adapter. The adapter injects the received wireless signal into the internal feature space of the generative model, enabling direct control of the generation trajectory and fine-grained manipulation of the synthesized content.
  \item \textbf{Efficient Training and Sampling}. We derive a tractable training objective using variational inference and Lagrangian relaxation, enabling end-to-end optimization of encoding, transmission, and generation. An efficient conditional sampling strategy is further developed for resource-constrained communication scenarios.
  \item \textbf{Theoretical Guarantees}. We theoretically prove that the derived training objective is consistent with the original information-theoretic formulation, promoting both mutual information maximization and perceptual consistency. We also illustrate that the efficient sampling process preserves high perceptual quality both in KL divergence and Wasserstein distance.
  \item \textbf{Numerical Experiments}. Extensive experiments show that the proposed JSCGC outperforms representative JSCC and generative communication baselines in terms of perceptual quality and semantic consistency under various channel conditions. For example, under the additive white Gaussian noise (AWGN) channel setting with a signal-to-noise ratio (SNR) of 5 dB on the Kodak dataset, JSCGC reduces LPIPS and FID to 79.42\% and 53.68\% of the baseline values achieved by DiffJSCC, respectively, while improving the CLIP score by 11\%.
\end{itemize}

The rest of this paper is organized as follows. The framework of the JSCGC is introduced in Section \ref{sec:framework}. The training and sampling algorithms of the JSCGC are presented in Section \ref{sec:algorithm}. Finally, extensive experimental results are presented in Section \ref{sec:experiments}, and conclusions are drawn in Section \ref{sec:conclusion}.
\section{Proposed JSCGC Framework}\label{sec:framework}
In this section, we first illustrate the architecture and objective of the JSCGC paradigm.
Then we give an example of the specific implementation for wireless image transmission. %

\subsection{Architecture and Objective of the JSCGC Paradigm}
\begin{figure*}[t]
  \centering
  \includegraphics[width=0.95\textwidth]{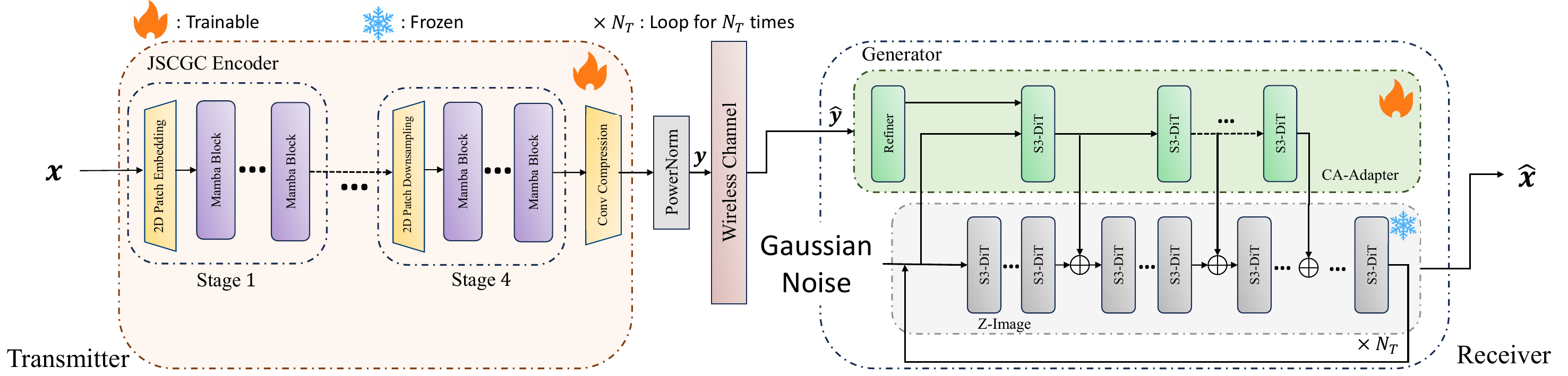}
  \caption{An implementation example of the proposed JSCGC scheme for wireless image transmission.} 
  \label{JSCGC_model}
  \end{figure*}
As shown in Fig.~\ref{System_model}, we consider a communication system in which the transmitter encodes the source data $\mathbf{x}\in \mathbb{R}^m$ into a channel signal $\mathbf{y}\in\mathbb{C}^l$ through an encoder $E_\phi(\cdot)$ parameterized by $\phi$, where $l<m$. The encoded signal $\mathbf{y}$ is then transmitted through a communication channel, resulting in the received signal $\hat{\mathbf{y}} \in \mathbb{C}^l$. 
At the receiver, unlike traditional JSCC schemes that perform deterministic reconstruction, the conventional decoder is replaced by a generative model. Specifically, the received signal $\hat{\mathbf{y}}$ serves as a semantic condition that guides conditional generation, and the generated output $\hat{\mathbf{x}}\in\mathbb{R}^{m}$ is sampled from $q_\theta(\mathbf{\hat{x}}|\hat{\mathbf{y}})$ parameterized by a generative model $G_{\theta}(\cdot)$. Noticing that $\hat{\mathbf{x}}$ and $\mathbf{x}$ lie in the same space, we equivalently write the generative distribution as $q_\theta(\mathbf{{x}}|\hat{\mathbf{y}})$ and $q_\theta(\mathbf{\hat{x}})$ for measuring the distributional distance with the source data distribution $p(\mathbf{x}|\hat{\mathbf{y}})$ and $p(\mathbf{x})$. Accordingly, the communication process can be represented by the following Markov chain:
\begin{align}\label{CMC}
\mathbf{X} \rightarrow \mathbf{Y} \rightarrow \hat{\mathbf{Y}} \rightarrow \hat{\mathbf{X}}.
\end{align}
Here, $\mathbf{X},\  \mathbf{Y},\  \hat{\mathbf{Y}},$ and $\hat{\mathbf{X}}$ denote the random variables corresponding to the realizations $\mathbf{x},\  \mathbf{y},\  \hat{\mathbf{y}},$ and $\hat{\mathbf{x}}$, respectively.

As established, the fundamental objective of communication is the reliable transfer of information. Meanwhile, the generated output is ultimately intended for human perception. Therefore, communication fidelity should be evaluated not only by information preservation but also by perceptual consistency with the natural data distribution. To this end, we impose a perceptual constraint requiring the distribution of the generated outputs to remain close to the distribution $p(\mathbf{x})$.

Accordingly, the design objective of JSCGC is formulated as the following constrained optimization problem:
\begin{align}\label{training_objective}
  &\max_{\theta, \phi}\  I(\mathbf{X}; \hat{\mathbf{Y}}), \nonumber\\
  &\text{s.t.}\quad d_p\big(p(\mathbf{x}), q_\theta(\mathbf{x})) \leq \zeta.
\end{align}
where $d_p(\cdot)$ denotes a divergence measure that quantifies the perceptual discrepancy between the source distribution and the generated data distribution, and $\zeta$ is a prescribed tolerance.

Unlike conventional RD and RDP formulations, which explicitly optimize distortion based objectives, JSCGC is formulated directly from an information preservation perspective. The communication system is optimized to maximize the information conveyed through the received signal while satisfying a perceptual consistency constraint, without relying on predefined distortion functions.

Although the above formulation provides a principled objective for generative communication, directly optimizing it is challenging because both the mutual information term and the perceptual divergence are generally intractable for high-dimensional data distributions. To address this challenge, we develop a practical JSCGC framework based on deep neural networks, through which the objective can be efficiently approximated and optimized in an end-to-end way.


\subsection{Implementation example of JSCGC}

As an example of the proposed paradigm, we implement JSCGC for wireless image transmission, as illustrated in Fig.~\ref{JSCGC_model}.
In this implementation, image transmission is performed in the latent space, where $\mathbf{x}$ denotes the latent representation of the source image extracted by a pre-trained extractor. The encoder is built based on the MambaJSCC architecture~\cite{MambaJSCCWu}, which employs multi-stage Mamba blocks to efficiently extract high-level semantic features and map them to the channel input for transmission. At the receiver, we adopt Z-Image \cite{z-image} as the generator to achieve strong generative capability with reduced computational complexity. Z-Image is a latent-space generative model built upon S3-DiT and pretrained on large-scale datasets, enabling it to effectively model the distribution of natural data.

However, directly incorporating communication signals $\hat{\mathbf{y}}$ into a large-scale pretrained generative model is challenging due to the mismatch between communication representations and generative feature spaces. To address this issue while preserving the pretrained knowledge of the generator, we design a communication-aware adapter (CA-Adapter) to modulate the generation process using the received signal. Specifically, the CA-Adapter consists of $N_I$ cascaded stages. At each stage, the adapter takes the received signal $\hat{\mathbf{y}}$ and the intermediate latent variable $\mathbf{x}_t$ as inputs. The resulting feature maps from the $i$-th adapter stage are injected into the $L_i$-th layer of the Z-Image backbone via element-wise addition, thereby allowing the received signal to directly influence the generation trajectory and enabling fine-grained control over the generated content.

\section{Training and Sampling Algorithms of JSCGC}\label{sec:algorithm}
In this section, we first introduce the training algorithm of the proposed JSCGC framework and analyze its consistency with the original communication objective. We then present the fast sampling algorithm and theoretical analysis its effectiveness in generating samples with low perception distance.
\subsection{Training Algorithm of JSCGC}
To obtain a tractable realization of the JSCGC objective in (\ref{training_objective}), we instantiate the perceptual discrepancy using the KL divergence and apply a Lagrangian relaxation to convert the constrained optimization into an unconstrained form.

Since the source distribution \(p(\mathbf{X})\) is fixed, maximizing $I(\mathbf{X}; \hat{\mathbf{Y}})$ is equivalent to minimizing the conditional entropy $H(\mathbf{X} | \hat{\mathbf{Y}})=\mathbb{E}_{(\mathbf{x},\hat{\mathbf{y}})\sim p(\mathbf{x},\hat{\mathbf{y}})}\left[-\log p(\mathbf{x}|\hat{\mathbf{y}})\right]$. 
Therefore, the optimization objective can be relaxed as
\begin{align}
\min_{\theta, \phi} H(\mathbf{X} | \hat{\mathbf{Y}}) + D_{KL}(p(\mathbf{x})||q_\theta(\mathbf{x})).
\end{align}

In practice, the true posterior distribution $p(\mathbf{x}|\hat{\mathbf{y}})$ requires knowledge of the ground-truth joint distribution of natural data and channel outputs, which is theoretically intractable. To address this issue, we have the following proposition.
\begin{proposition}\label{prop:variational_info}
The objective $H(\mathbf{X} | \hat{\mathbf{Y}}) + D_{KL}(p(\mathbf{x})||q_\theta(\mathbf{x}))$ can be minimized via minimizing the loss function:
\begin{align}\label{eq:gradient_loss}
\mathcal{L}(\theta,\phi)= \mathbb{E}_{t, \mathbf{x}_0, \hat{\mathbf{y}}, \mathbf{\epsilon}} \left[ \| (\mathbf{\epsilon} - \mathbf{x}_0) - \boldsymbol{\nu}_\theta(\mathbf{x}_t, t, \hat{\mathbf{y}}) \|^2 \right].
\end{align}
where $\boldsymbol{\nu}_\theta(\cdot)$ is the output of the neural network in the generator, $\mathbf{x}_0=\mathbf{x}, \mathbf{\epsilon} \sim N(0,\mathbf{I})$, and
\begin{align}\label{forward_process}
\mathbf{x}_t=(1-\frac{t}{T})\mathbf{x}_0+\frac{t}{T} \mathbf{\epsilon} \quad \text{ for } t\in\{1,2,...,T\},
\end{align}
\end{proposition}
\begin{proof}
  The proof is provided in the Appendix \ref{app:training_loss}.
  \end{proof}

\begin{remark}[Joint Source-Channel-Generation Learning]
  The proposed training procedure naturally couples source coding, channel transmission, and generative modeling into a unified optimization framework. As shown by the gradient with respect to the encoder parameters,
\begin{equation}
  \begin{aligned}
    \nabla_\phi \mathcal{L} = -2 \mathbb{E}_{t, \mathbf{x}_0, \mathbf{\epsilon}} \Bigg[ \underbrace{\left( (\mathbf{\epsilon} - \mathbf{x}_0) - \boldsymbol{\nu}_\theta(\mathbf{x}_t, t, \hat{\mathbf{y}}) \right)}_{\text{Estimation Error}} \\ \cdot \underbrace{\nabla_{\hat{\mathbf{y}}} \boldsymbol{\nu}_\theta(\mathbf{x}_t, t, \hat{\mathbf{y}})}_{\text{Generation Term}} \cdot \underbrace{\nabla_{\mathbf{y}} W(\mathbf{y})}_{\text{Channel Term}} \cdot \underbrace{\nabla_\phi E_\phi(\mathbf{x}_0)}_{\text{Source Term}} \Bigg].
  \end{aligned}
  \end{equation}
The encoder is optimized through gradients propagated from the generative model via the communication channel. This results in a fully joint optimization of source, channel, and generation models, distinguishing JSCGC from existing generative communication schemes where the encoder is typically independent of the generator.
\end{remark}
\begin{remark}\label{paradigmshift}
Compared with conventional reconstruction-based JSCC and rate–distortion–perception formulations, the proposed framework removes the need for explicit distortion functions and deterministic decoders. Instead, communication is reformulated as controllable conditional generation guided by the received signal.
\end{remark}

Overall, the proposed training algorithm provides a tractable variational realization of the JSCGC objective (\ref{training_objective}). It simultaneously minimizes a surrogate of the perceptual divergence and maximizes a lower bound of mutual information, thereby aligning learning with the original information-theoretic formulation.

\subsection{Theoretical Analysis of Training Effectiveness}

We now analyze whether optimizing the proposed surrogate loss in \eqref{eq:gradient_loss} is consistent with the original JSCGC objective in \eqref{training_objective}. Specifically, we examine its impact on both the perceptual constraint and the mutual information term.

\begin{proposition}[Optimal Solution of the Training Objective]\label{prop:optimal_solution}
  The optimal solution $\boldsymbol{\nu}^*(\mathbf{x}_t, t, \hat{\mathbf{y}})$ of the training objective of $\mathcal{L}(\theta,\phi) = \mathbb{E}_{t, \mathbf{x}_0, \hat{\mathbf{y}}, \mathbf{\epsilon}} \left[ \| (\mathbf{\epsilon} - \mathbf{x}_0) - \boldsymbol{\nu}_\theta(\mathbf{x}_t, t, \hat{\mathbf{y}}) \|^2 \right]$ is 
  \begin{equation}
\boldsymbol{\nu}^*(\mathbf{x}_t, t, \hat{\mathbf{y}}) =  - \frac{t\nabla_{\mathbf{x}_t} \log p_t(\mathbf{x}_t|\hat{\mathbf{y}})+T\mathbf{x_t}}{T-t} 
  \end{equation}
\end{proposition}
\begin{proof}
  The proof is shown in Appendix \ref{app:optimal_solution}
\end{proof}
Proposition \ref{prop:optimal_solution} shows that the neural network $\boldsymbol{\nu}_\theta(\cdot)$ is trained to approximate the theoretical optimum $\boldsymbol{\nu}^*(\cdot)$, which implicitly estimates $-\frac{t\nabla_{\mathbf{x}_t} \log p_t(\mathbf{x}_t|\hat{\mathbf{y}}) + T\mathbf{x}_t}{T-t}$.

We next characterize the distributions $p(\mathbf{x})$ and $q_\theta(\mathbf{x})$. For notational clarity, the learned distribution induced by the reverse-time stochastic differential equation (SDE) is denoted by $q_\theta^s(\mathbf{x})$.
\begin{corollary}\label{cor:SDE_derivation}
  The true data distribution $p(\mathbf{x})$ and the learned distribution $q_\theta^s(\mathbf{x})$ can be expressed as:
  \begin{align}
    p(\mathbf{x}) = \int p(\mathbf{x}|\hat{\mathbf{y}}) p(\hat{\mathbf{y}}) d\hat{\mathbf{y}}, \quad q_\theta^s(\mathbf{x}) = \int q_\theta^s(\mathbf{x}|\hat{\mathbf{y}}) p(\hat{\mathbf{y}}) d\hat{\mathbf{y}}.
  \end{align}
  where $p(\mathbf{x}|\hat{\mathbf{y}})$ and $q_\theta^s(\mathbf{x}|\hat{\mathbf{y}})$ are the marginal distributions of the SDEs at $t=0$, initialized with $\mathbf{x}_T\sim \pi(\mathbf{x}_T)$
\begin{align}
    &d\mathbf{x}_t = [\frac{2}{T}\boldsymbol{\nu}^*(\mathbf{x}_t, t,\hat{\mathbf{y}}) - \mathbf{f}(\mathbf{x}_t, t)] dt + g(t) d\bar{\mathbf{w}}_t,\label{eq:true_sde}\\
    &d\mathbf{x}_t = [\frac{2}{T}\boldsymbol{\nu}_\theta(\mathbf{x}_t, t,\hat{\mathbf{y}}) - \mathbf{f}(\mathbf{x}_t, t)] dt + g(t) d\bar{\mathbf{w}}_t,\label{eq:learned_sde}
\end{align}
respectively, where $\mathbf{f}(\mathbf{x}_t, t)= -\frac{1}{T-t}\mathbf{x}_t$, $g(t)= \sqrt{\frac{2t}{T(T-t)}}$ and $\pi(\mathbf{x}_T)$ is the standard Gaussian distribution.
\end{corollary}
\begin{proof}
  According to the forward process defined in (\ref{forward_process}), the discrete process converges to the following SDE as the number of sampling steps approaches infinity \cite{Score-DM}:
  \begin{equation}
    d\mathbf{x}_t = \mathbf{f}(\mathbf{x}_t, t)dt + g(t) d\mathbf{w}_t, \quad  \mathbf{x}_0 \sim p(\mathbf{x}|\hat{\mathbf{y}}),
  \end{equation}
where marginal distribution at $t=0$ is $p(\mathbf{x}|\hat{\mathbf{y}})$. According to Anderson's theorem, its time reversal SDE is 
  \begin{align}
   d\mathbf{x}_t = [\mathbf{f}(\mathbf{x}_t, t)-g(t)^2\nabla_{\mathbf{x}_t}\log p_t(\mathbf{x}_t|\hat{\mathbf{y}})] dt + g(t) d\bar{\mathbf{w}}_t.
  \end{align}
  Substituting the result in Proposition \ref{prop:optimal_solution} into the above equation yields (\ref{eq:true_sde}).
  Therefore, $p(\mathbf{x}|\hat{\mathbf{y}})$ is also the marginal distribution of the forward SDE \eqref{eq:true_sde} at $t=0$.
  Furthermore, the parameterized network $\boldsymbol{\nu}_\theta(\mathbf{x}_t, t, \hat{\mathbf{y}})$ is trained to approximate the optimal solution $\boldsymbol{\nu}^*(\mathbf{x}_t, t, \hat{\mathbf{y}})$. Therefore, $q_\theta^s(\mathbf{x}|\hat{\mathbf{y}})$ is given by the marginal distribution at $t=0$ of the SDE in \eqref{eq:learned_sde}.
\end{proof}

With the above two clarifications, we now establish the theoretical relationship between $D_{KL}(p(\mathbf{x})||q_\theta(\mathbf{x}))$ with $q_\theta(\mathbf{x})=q_\theta^s(\mathbf{x})$ and the training loss in (\ref{eq:gradient_loss}), in order to characterize the effectiveness of the training algorithm in minimizing the original objective in (\ref{training_objective}).
\begin{proposition}[Upper Bound of the KL divergence in the Training Process]\label{SDEKLbound}
  With the proposed  JSCGC training algorithm, the KL divergence $D_{KL}(p(\mathbf{x})||q_\theta^s(\mathbf{x}))$ is upper bounded by:
\[
\int_0^T \frac{2}{g(t)^2 T^2} \mathbb{E}_{p_t(\mathbf{x}_t, \hat{\mathbf{y}})} \left[ \big\| \boldsymbol{\nu}^*(\mathbf{x}_t, t, \hat{\mathbf{y}}) - \boldsymbol{\nu}_\theta(\mathbf{x}_t, t, \hat{\mathbf{y}}) \big\|^2 \right] dt
\]
\end{proposition}

\begin{proof}
  The proof is shown in Appendix \ref{app:SDEKLbound}.
\end{proof}

The above proposition illustrates that decreasing the proposed loss reduces an upper bound on the KL divergence. According to Proposition \ref{prop:optimal_solution}, $\boldsymbol{\nu}^*(\cdot)$ is the optimal solution of $\boldsymbol{\nu}_\theta(\cdot)$. As a result, the upper bound approaches zero when the optimal solution is achieved. Therefore, minimizing the proposed loss effectively drives the minimization of the perception distance in (\ref{training_objective}) when the perceptual constraint is measured by the KL divergence.

We next analyze the relationship between maximizing the mutual information $I(\mathbf{X};\hat{\mathbf{Y}})$ and minimizing the loss function in (\ref{eq:gradient_loss}). From Proposition \ref{prop:variational_info}, minimizing the training loss is equivalent to minimizing an evidence upper bound of $\mathbb{E}_{p(\mathbf{x}_0,\hat{\mathbf{y}})}\left[-\log q_\theta(\mathbf{x}_0|\hat{\mathbf{y}})\right]$, where the upper bound is given by $\mathbb{E}_{p(\mathbf{x}_{0:T},\hat{\mathbf{y}})}\left[-\log \frac{q_\theta(\mathbf{x}_{0:T}|\hat{\mathbf{y}})}{p(\mathbf{x}_{1:T}|\mathbf{x}_0, \hat{\mathbf{y}})}\right]$.
Therefore, we obtain
\begin{align}
I(\mathbf{X};\hat{\mathbf{Y}}) \geq H(\mathbf{X})-\mathbb{E}_{p(\mathbf{x}_{0:T},\hat{\mathbf{y}})}\left[-\log \frac{q_\theta(\mathbf{x}_{0:T}|\hat{\mathbf{y}})}{p(\mathbf{x}_{1:T}|\mathbf{x}_0, \hat{\mathbf{y}})}\right].
\end{align}
Since the source entropy $H(\mathbf{X})$ is fixed, minimizing the loss in (\ref{eq:gradient_loss}) effectively maximizes a tractable lower bound on the mutual information $I(\mathbf{X};\hat{\mathbf{Y}})$. Furthermore, in the proposed framework, $\hat{\mathbf{y}}$ is encoded by $E_\phi(\cdot)$ and transmitted through the wireless channel. Since $I(\mathbf{X};\hat{\mathbf{Y}})$ is generally close to zero under randomly initialized encoder parameters $\phi$, maximizing this lower bound encourages the encoder to preserve more information about $\mathbf{X}$ in $\hat{\mathbf{Y}}$.

\subsection{Sampling Algorithm of the JSCGC}
After training, the receiver generates samples conditioned on the received signal $\hat{\mathbf{y}}$, which carries information about the source signal $\mathbf{x}$. Owing to the proposed JSCGC framework, JSCGC naturally learns a controllable generator characterized by the learned SDE in (\ref{eq:learned_sde}). As shown in Proposition \ref{SDEKLbound}, the generated samples can achieve high perceptual quality by approximating the target data distribution. However, sampling from the SDE remains computationally expensive due to its stochastic nature and typically requires a large number of discretization steps (e.g., 1000 steps), resulting in substantial latency. Fortunately, it has been shown in \cite{Score-DM} that each SDE is associated with a probability flow ODE that shares the same marginal distributions at all time instances. Therefore, instead of simulating the stochastic process, samples can be generated by solving the corresponding ODE using efficient numerical solvers with significantly fewer steps.

According to the theorem of SDE, with the definitions of $\mathbf f(\mathbf{x}_t,t)$ and $g(t)$, the true ODE of the true SDE in \eqref{eq:true_sde} is
\begin{align}\label{eq:true_ode}
  \frac{d\mathbf{x}_t}{dt}=\frac{1}{T}\boldsymbol{\nu}^*(\mathbf{x}_t,t,\hat{\mathbf y}),\qquad\mathbf{x}_T\sim \pi(\mathbf{x}_T).
\end{align}

Since JSCGC learns the optimal velocity field $\boldsymbol{\nu}^*(\mathbf{x}_t, t, \hat{\mathbf{y}})$ through its neural approximation $\boldsymbol{\nu}_\theta(\mathbf{x}_t, t, \hat{\mathbf{y}})$, the practical sampling process is performed using the learned ODE
\begin{equation}\label{eq:learn_ode}
   \frac{d\mathbf{x}_t}{dt} = \frac{\boldsymbol{\nu}_\theta(\mathbf{x}_t, t, \hat{\mathbf{y}})}{T},\quad \mathbf{x}_T \sim \pi(\mathbf{x}_T).
\end{equation}
The learned ODE can be efficiently solved using numerical ODE solvers, i.e., Euler method with finite iterative steps. 
\begin{remark}
The marginal equivalence between the ODE and the SDE relies on the exact score function of the underlying data distribution. In Proposition \ref{prop:optimal_solution}, this score information is implicitly characterized through the optimal $\boldsymbol{\nu}^*(\mathbf{x}_t,t,\hat{\mathbf y})$. Therefore, the probability flow ODE corresponding to the true reverse-time SDE can be rigorously established. In practice, however, only the learned approximation $\boldsymbol{\nu}_{\theta}(\mathbf{x}_t,t,\hat{\mathbf y})$ is available. Therefore, the practical sampling ODE in (\ref{eq:learn_ode}) should be viewed as an approximation to the true probability flow ODE, whose accuracy depends on how closely $\boldsymbol{\nu}_{\theta}$ approximates $\boldsymbol{\nu}^{*}$.
\end{remark}

\subsection{Theoretical Analysis of the Sampling Effectiveness}
To improve sampling efficiency, we replace the reverse process based on the SDE in \eqref{eq:learned_sde} with an ODE in \eqref{eq:learn_ode}. However, since $\boldsymbol{\nu}_\theta(\mathbf{x}_t,t,\hat{\mathbf y})$ is only an approximation to $\boldsymbol{\nu}^*(\mathbf{x}_t,t,\hat{\mathbf y})$, the learned ODE does not exactly match the true ODE. The distribution induced by the ODE does not coincide with the distribution derived during training by the SDE. Therefore, given the distributional constraint specified in the system design objective \eqref{training_objective}, it is necessary to further analyze the distance between the sample distribution produced by the ODE and the original distribution. Here, the sampling distribution $q_\theta(\mathbf{x})$ from the ODE is denoted as $q_\theta^o(\mathbf{x})$ to distinguish it from the SDE sampling distribution $q_\theta^s(\mathbf{x})$.

The true conditional marginal density $p_t(\mathbf{x}_t\mid \hat{\mathbf y})$ and the learned conditional marginal density $q_{t,\theta}^{o}(\mathbf{x}_t\mid \hat{\mathbf y})$ are induced at time $t$ by the true ODEs in \eqref{eq:true_ode} and the learned ODE in \eqref{eq:learn_ode}, respectively. 
The true data distribution $p(\mathbf{x})$ and the learned data distribution $q_\theta^o(\mathbf{x})$ can be derived as
\[
    p(\mathbf{x}) = \int p_0(\mathbf{x}_0|\hat{\mathbf{y}}) p(\hat{\mathbf{y}}) d\hat{\mathbf{y}}, \quad q_\theta^o(\mathbf{x}) = \int q_{0,\theta}^o(\mathbf{x}_0|\hat{\mathbf{y}}) p(\hat{\mathbf{y}}) d\hat{\mathbf{y}}.
\]

We now analyze the distance between the true and learned distributions in terms of the KL divergence and the Wasserstein distance.
\begin{proposition}[ODE-induced KL divergence dynamics]\label{ODEKLbound}
Under the ODE sampling, the KL divergence $D_{KL}(p(\mathbf{x})||q_\theta^o(\mathbf{x}))$ is upper-bounded by
\begin{align}\label{eq:ode_kl_dynamics_residual}
&\int_{0}^{T}\frac{T-t}{Tt}\mathbb{E}_{p_t(\mathbf{x}_t, \hat{\mathbf y})}\left[\left\|\boldsymbol{\nu}^*(\mathbf{x}_t,t,\hat{\mathbf y})-\boldsymbol{\nu}_\theta(\mathbf{x}_t,t,\hat{\mathbf y})\right\|^2\right]\\
&+\frac{1}{T}\mathbb{E}_{p_t(\mathbf{x}_t, \hat{\mathbf y})}\left[\left(\boldsymbol{\nu}^*(\mathbf{x}_t,t,\hat{\mathbf y})-\boldsymbol{\nu}_\theta(\mathbf{x}_t,t,\hat{\mathbf y})\right)^\top\mathbf r_\theta(\mathbf{x}_t,t,\hat{\mathbf y})\right]\,dt,\nonumber
\end{align}
where
$\mathbf r_\theta(\cdot):=\nabla_{\mathbf{x}_t}\log q_{t,\theta}^{o}(\mathbf{x}_t\mid \hat{\mathbf y})+\frac{T-t}{t}\boldsymbol{\nu}_\theta(\mathbf{x}_t,t,\hat{\mathbf y})+\frac{T}{t}\mathbf{x}_t.$
\end{proposition}

\begin{proof}
  The proof is shown in Appendix \ref{app:ODEKLbound}.
\end{proof}

The proposition shows that the upper bound of the KL divergence between the sample distribution induced by the ODE and the true data distribution can also be minimized by the training loss of JSCGC. Although it contains a residual term $\mathbf r_\theta(\cdot)$, we observe that when $\boldsymbol{\nu}_\theta(\cdot)=\boldsymbol{\nu}^*(\cdot)$, this upper bound equals zero. Therefore, sample generation via the ODE-based reverse process still preserves high perceptual quality.

Although the KL divergence cannot be directly obtained in this case, we can instead provide an upper bound for $d_p(p(\mathbf{x}),q_{\theta}^{o}(\mathbf{x}))$ in terms of the 2-Wasserstein distance ($W_2$) under the Lipschitz assumption. This allows us to quantify the extent to which the samples generated by the reverse ODE satisfy the distributional constraint specified in \eqref{training_objective}.
\begin{proposition}[$W_2$ distance bound for the ODE-induced marginal distributions]\label{ODEW2bound}
Assume that, uniformly over \(\hat{\mathbf y}\), the mapping
\(\mathbf x_t\mapsto \boldsymbol\nu_\theta(\mathbf x_t,t,\hat{\mathbf y})\)
is \(L_p(t)\)-Lipschitz. The $W_2$ distance between $p(\mathbf{x})$ and $q_{\theta}^{o}(\mathbf{x})$ satisfies
\begin{align}
W_2^2&\left(p(\mathbf{x}),q_{\theta}^{o}(\mathbf{x})\right)
\le\int_0^T\exp\!\left(\int_0^s \frac{2L_p(\tau)+1}{T}\,d\tau\right)\nonumber\\
&\cdot\frac{1}{T}\mathbb{E}_{p_s(\mathbf{x}_s, \hat{\mathbf y})}\left[\left\|\boldsymbol{\nu}^*(\mathbf{x}_s,s,\hat{\mathbf y})-\boldsymbol{\nu}_{\theta}(\mathbf{x}_s,s,\hat{\mathbf y})\right\|^2\right]ds.
\label{eq:w2_bound_conditional}
\end{align}
\end{proposition}

\begin{proof}
  The proof is shown in Appendix \ref{app:ODEW2bound}.
\end{proof}
The proposition indicates that, under ODE-based sampling, the $W_2$ distance between the generated sample distribution and the true data distribution is also upper-bounded and controlled by the training loss. Therefore, with sufficient training, the ODE-generated samples remain of high quality and satisfy the perceptual constraint.

\begin{remark}[Error Behavior Shifting]\label{change}
Propositions \ref{SDEKLbound}, \ref{ODEKLbound}, and \ref{ODEW2bound} reveal a fundamental difference between JSCGC and conventional distortion-oriented communication systems. In traditional reconstruction paradigms, performance degradation caused by limited communication resources is primarily reflected as increasing reconstruction distortion, such as blurring, artifacts, or structural degradation. In contrast, JSCGC explicitly constrains the generated distribution to remain close to the natural data distribution. Consequently, even under severe communication constraints, the generated samples can maintain high perceptual realism and visual quality.

As the amount of transmitted information decreases, the dominant performance degradation gradually shifts from distortion errors to semantic inconsistency. That is, the generated content may become less faithful to the original source while still appearing realistic and visually plausible. This fact represents a transition from distortion-limited communication to generation-controlled communication, and constitutes a distinctive characteristic of the proposed JSCGC paradigm.
\end{remark}

\section{Experimental Results}\label{sec:experiments}
\begin{figure*}[t]
  \centering
  \subfigure[]{\includegraphics[width=0.32\textwidth]{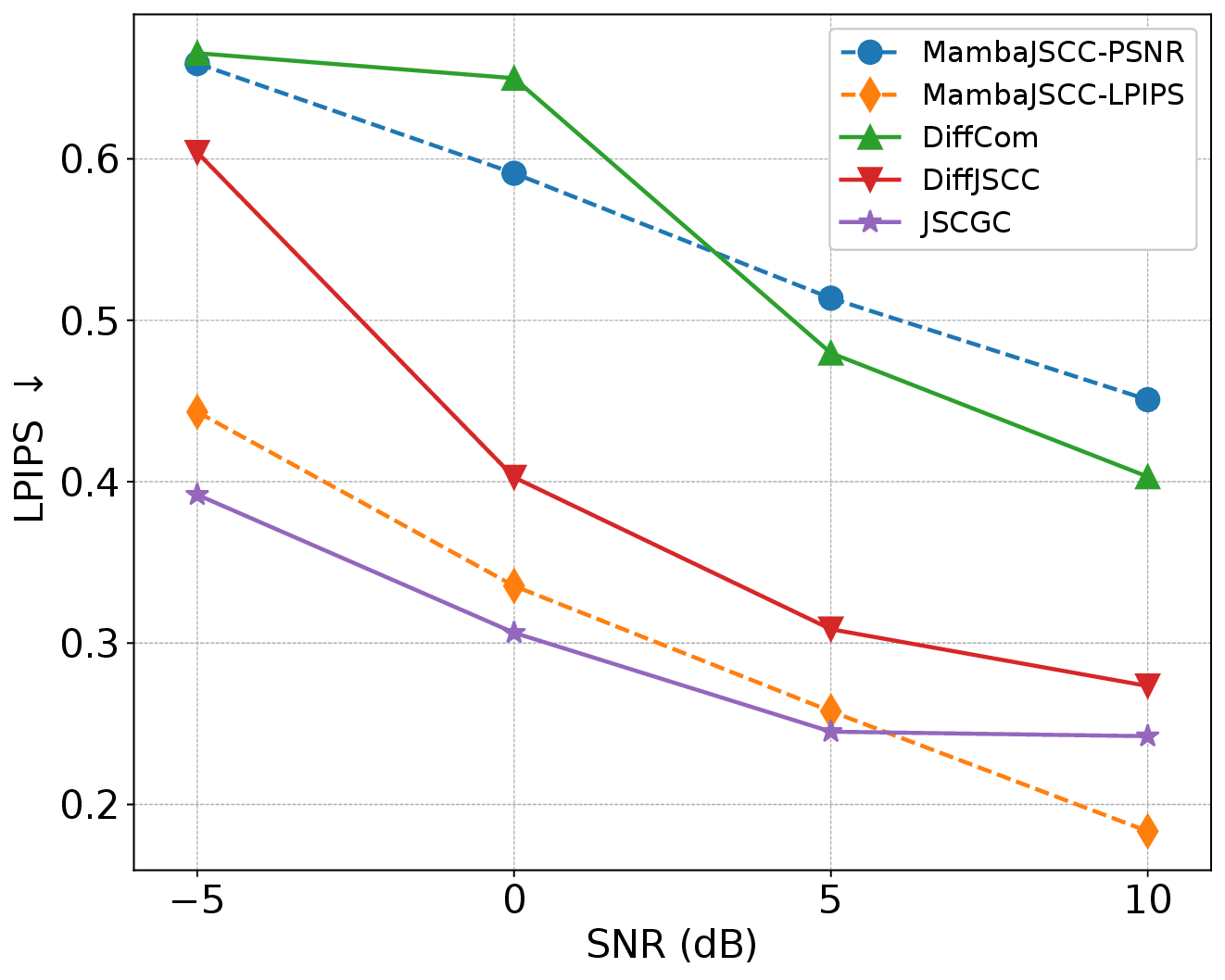}}
  \subfigure[]{\includegraphics[width=0.32\textwidth]{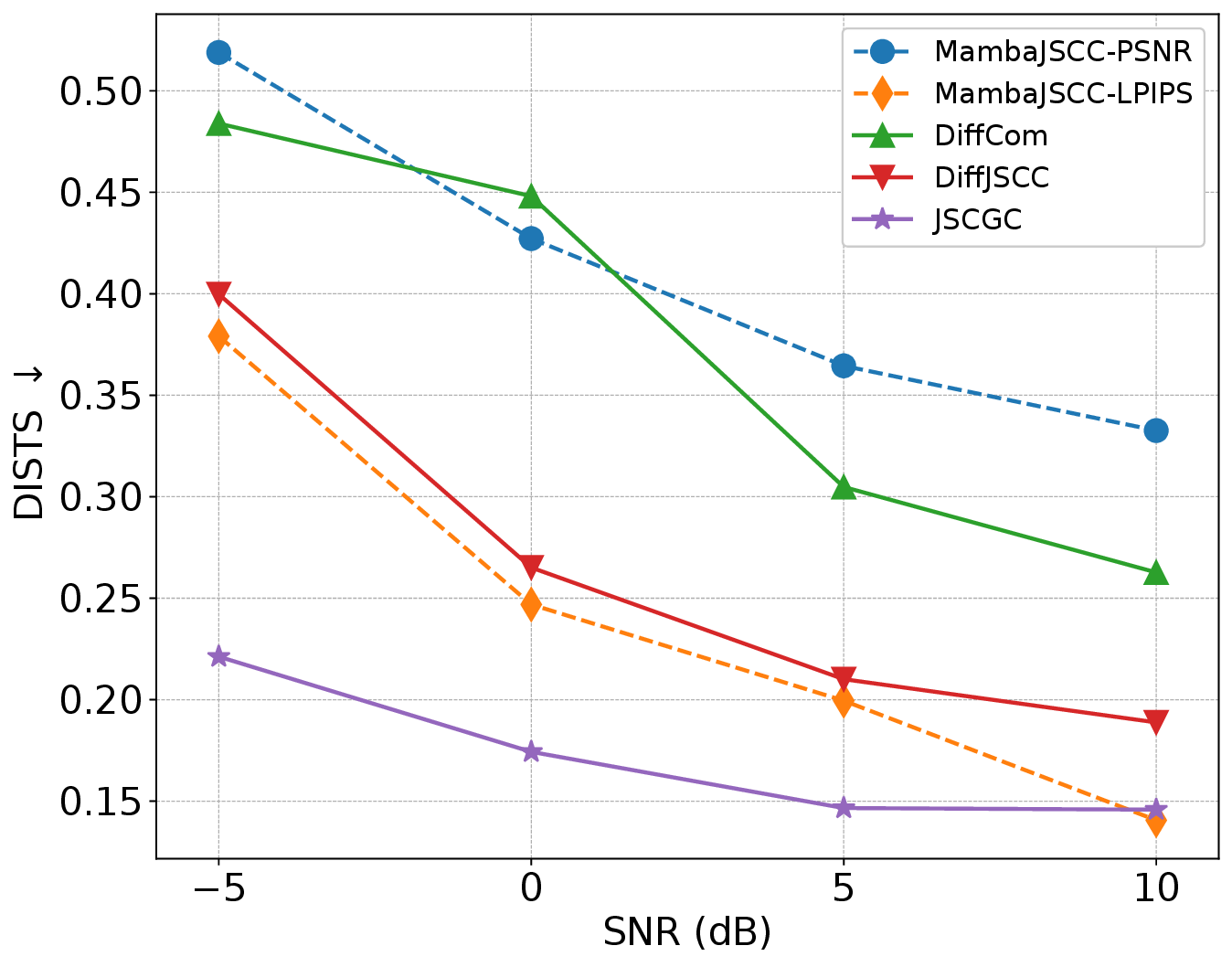}}
  \subfigure[]{\includegraphics[width=0.32\textwidth]{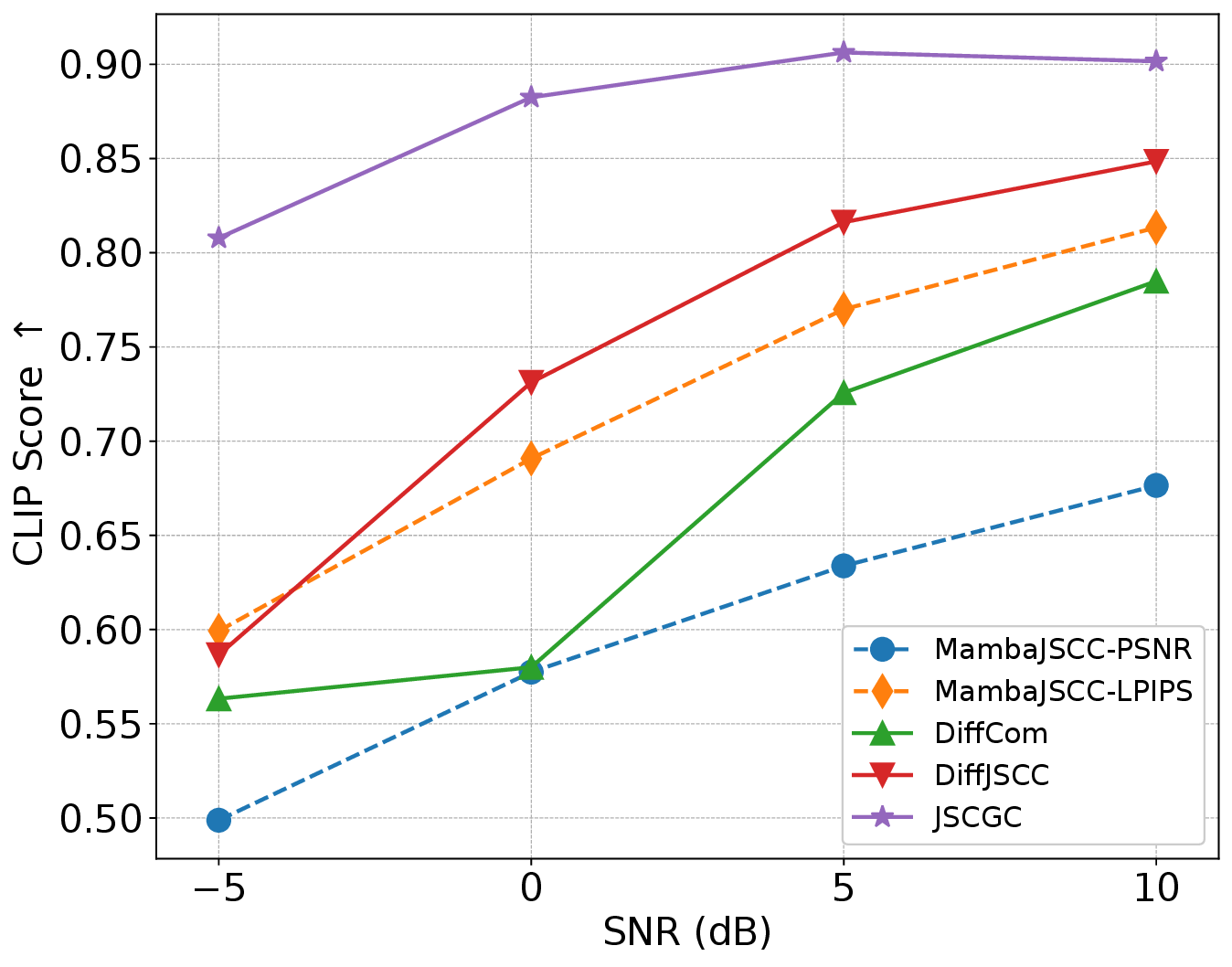}}
  \subfigure[]{\includegraphics[width=0.32\textwidth]{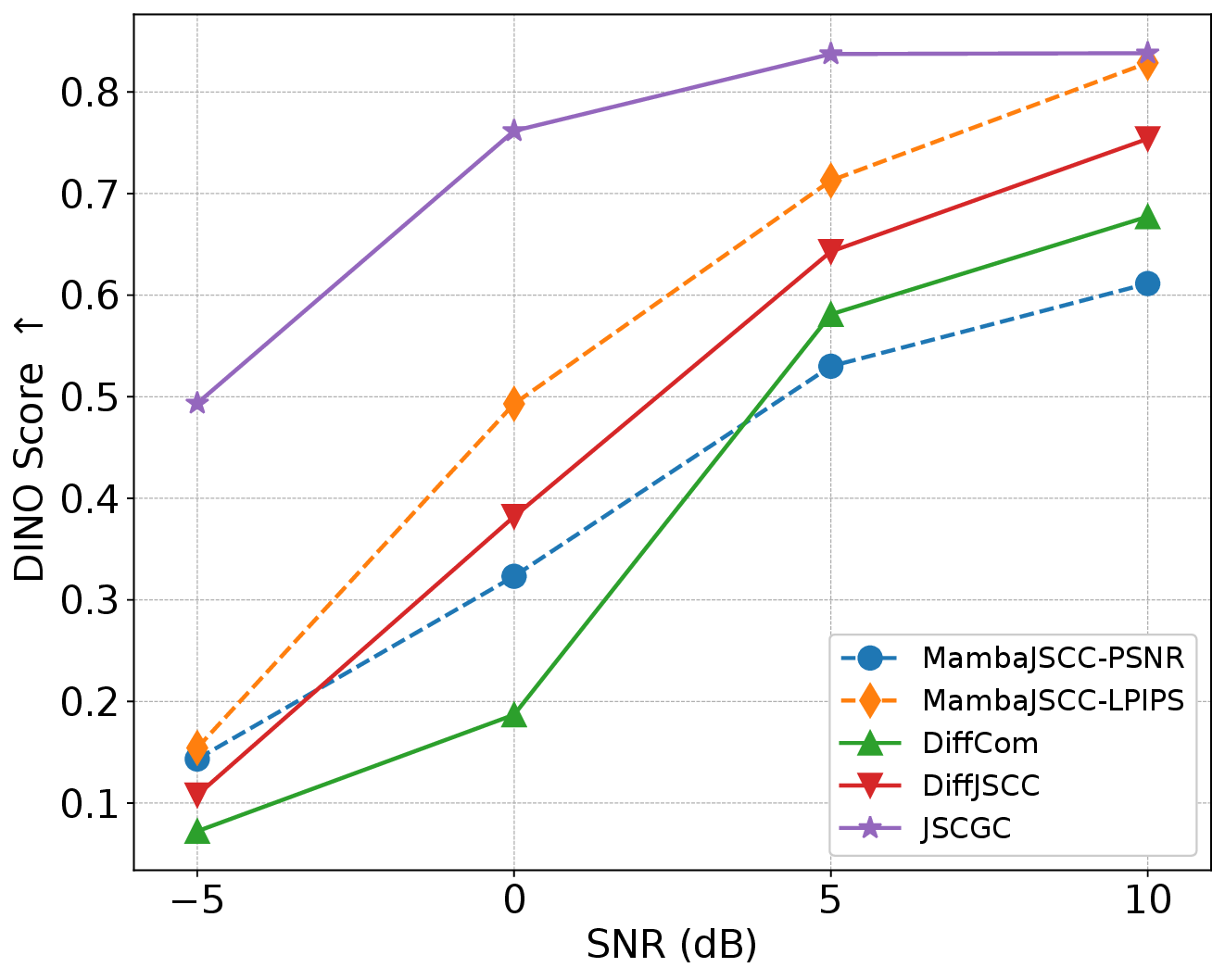}}
  \subfigure[]{\includegraphics[width=0.32\textwidth]{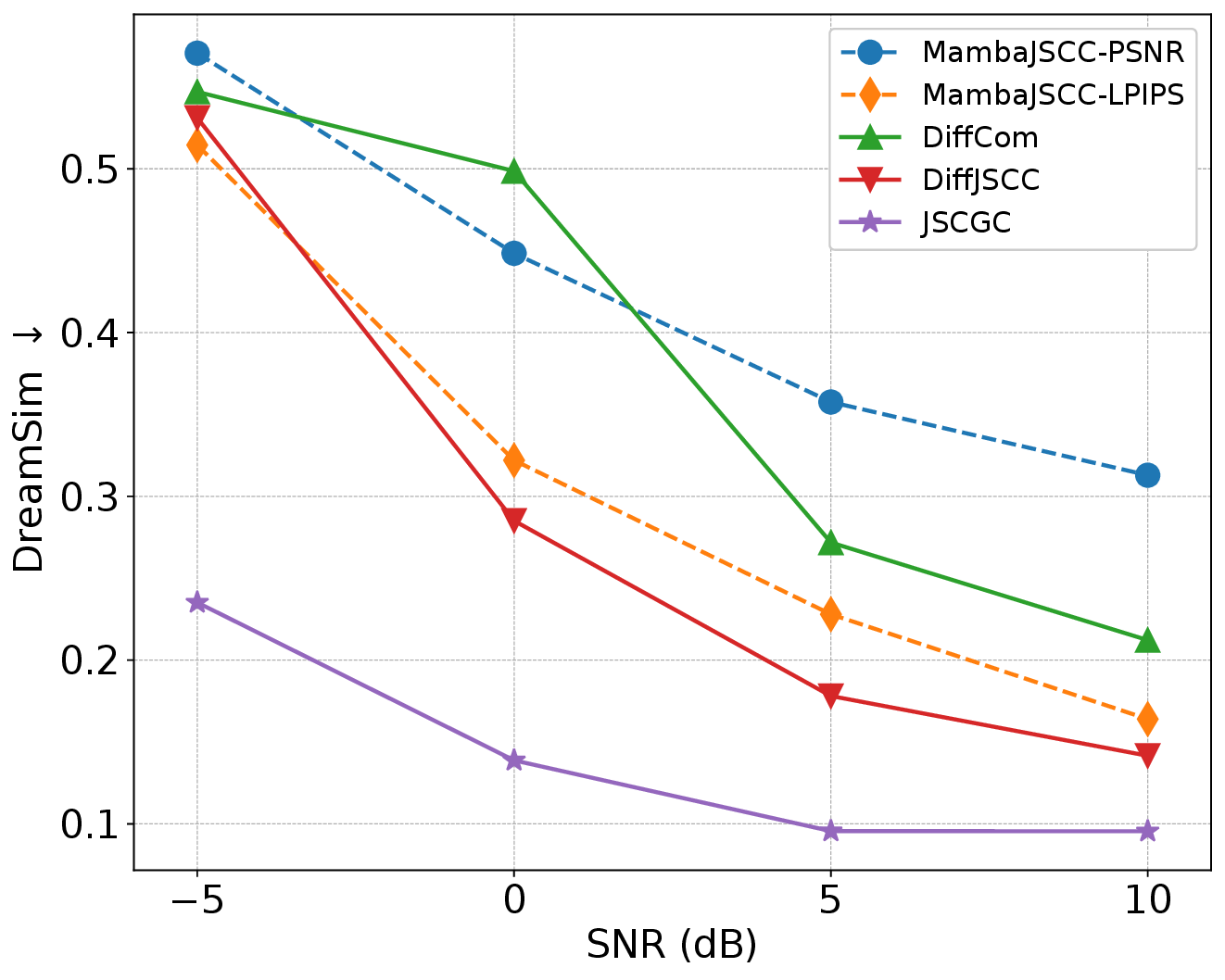}}
  \subfigure[]{\includegraphics[width=0.32\textwidth]{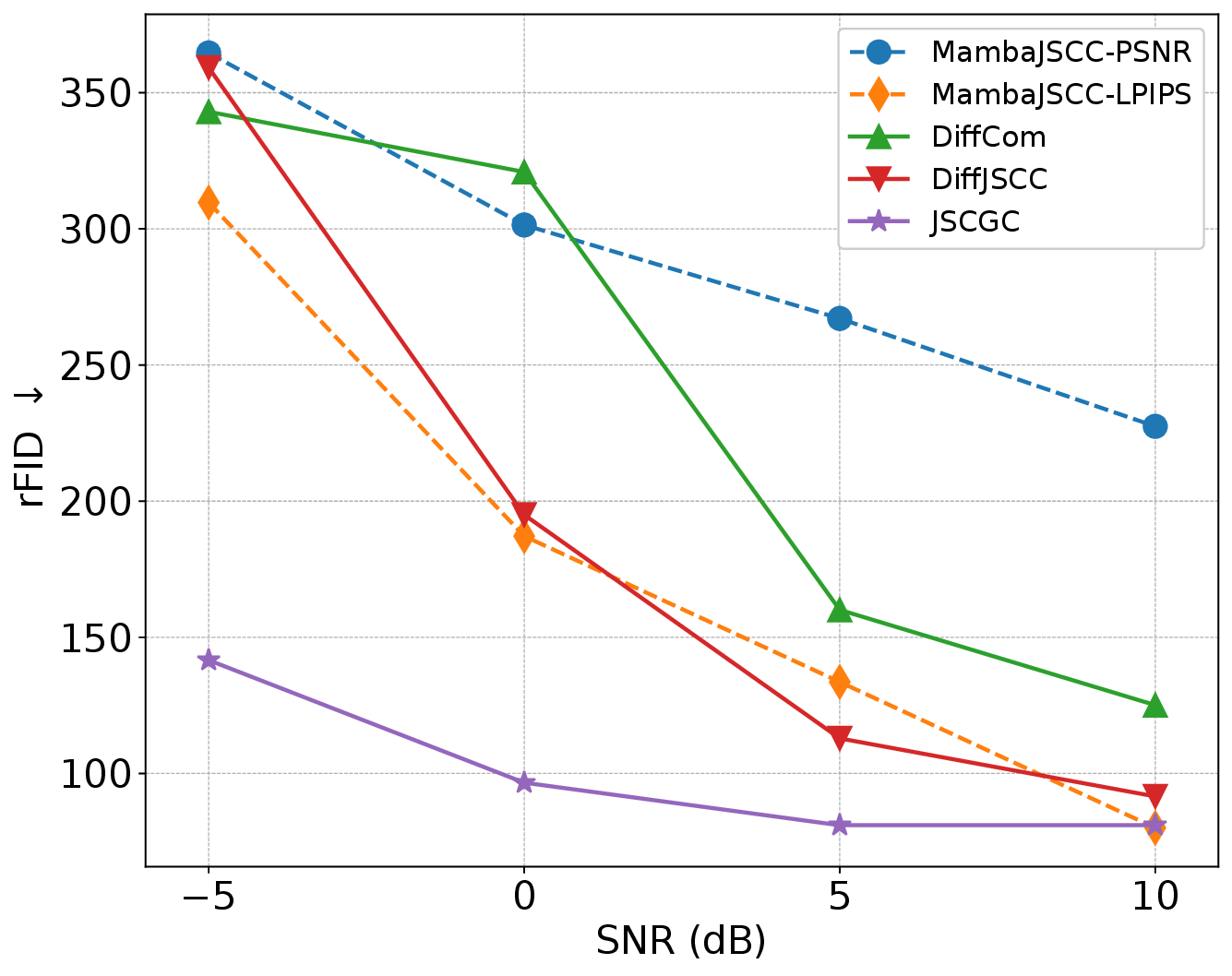}}
  \caption{The performance of different schemes versus SNR under the AWGN channel with CBR = $\frac{2}{768}$. (a) LPIPS. (b) DISTS. (c) CLIP Score. (d) DINO Score. (e) DreamSim. (f) rFID.} 
  \label{awgn_results}
  \vspace{-0.3cm}
\end{figure*}

\begin{figure}[t]
  \centering
  \includegraphics[width=0.48\textwidth]{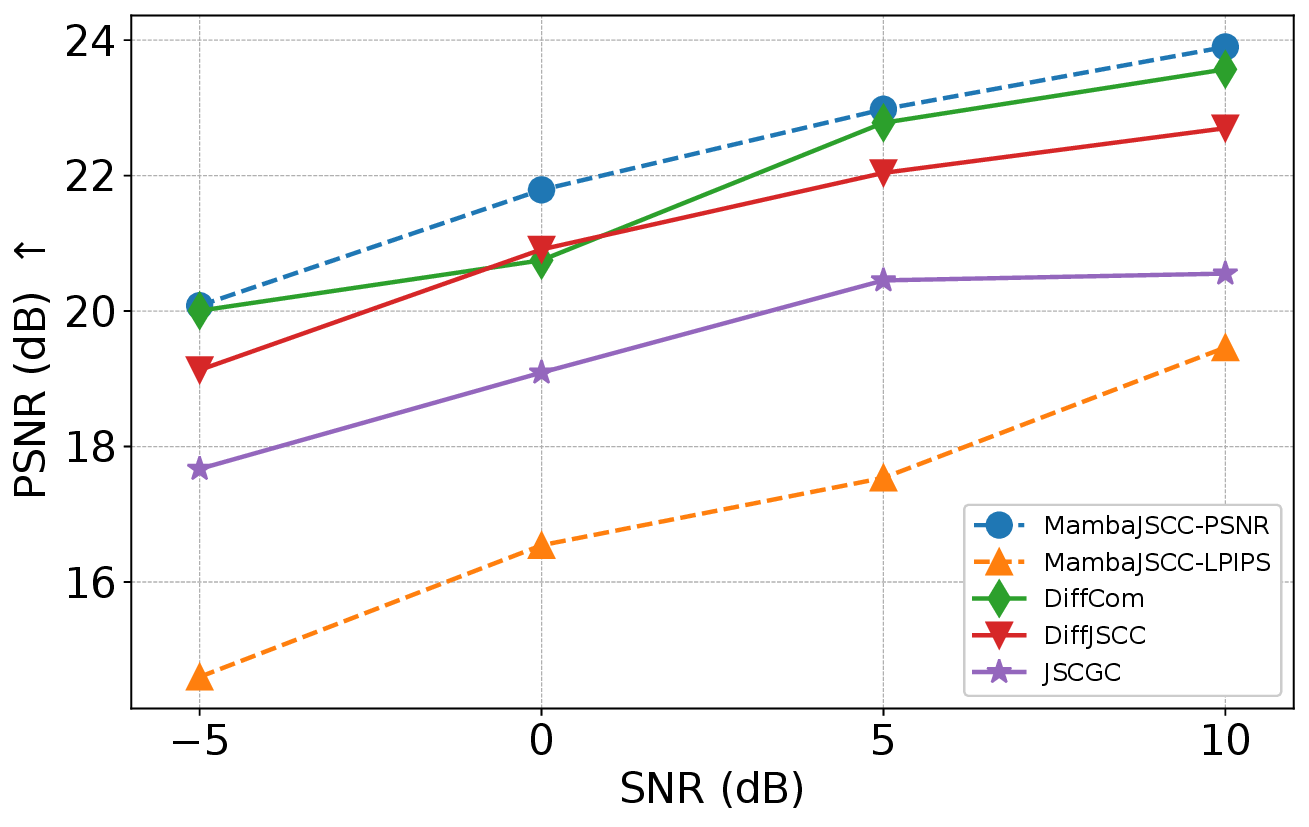}
  \caption{The PSNR performance of different schemes versus SNR under the AWGN channel with CBR = $\frac{2}{768}$.} 
  \label{PSNR_awgn_C4}
  \vspace{-0.5cm}
\end{figure}


\begin{figure*}[t]
  \centering
  \includegraphics[width=0.98\textwidth]{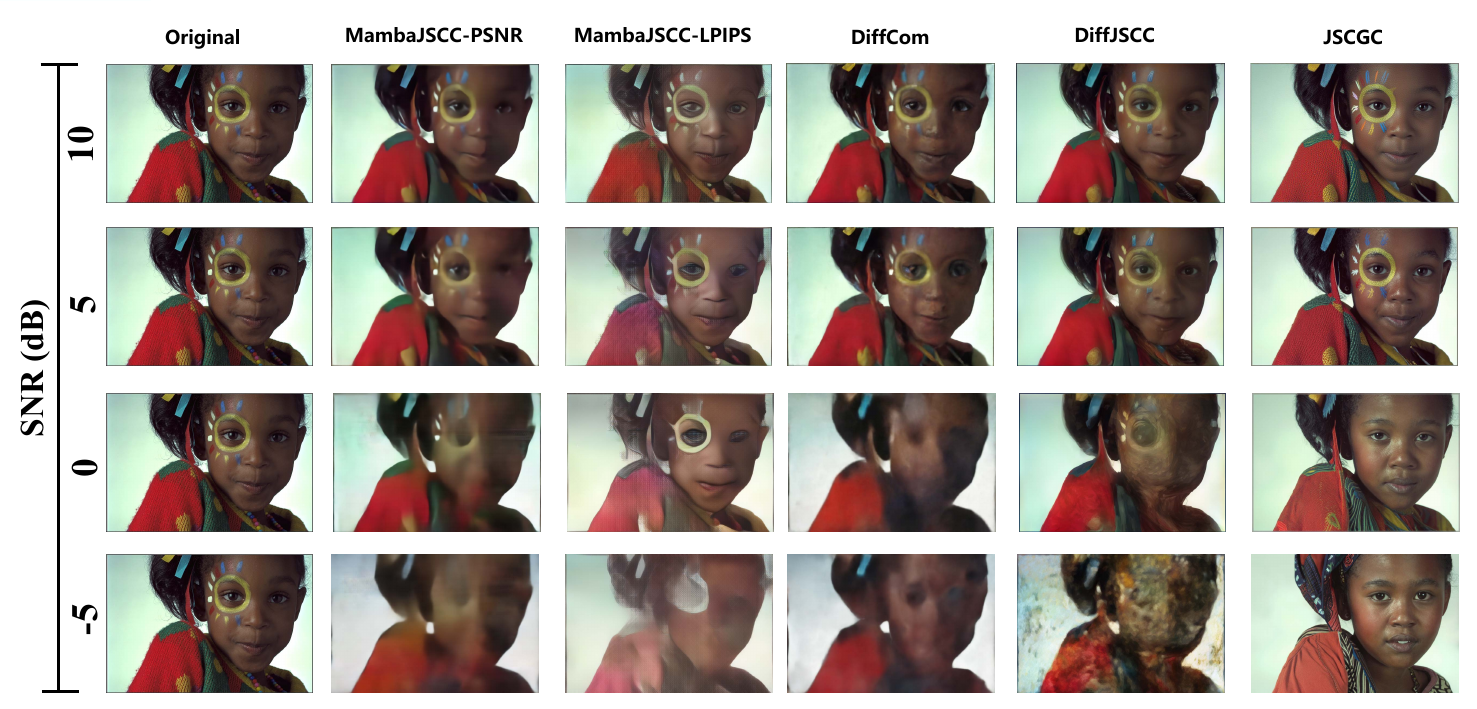}
  \caption{Visualization results under the AWGN channel with a CBR of $\frac{2}{768}$.} 
  \label{Visualizing1}
  \vspace{-0.3cm}
\end{figure*}

\begin{figure*}[t]
  \centering
  \subfigure[]{\includegraphics[width=0.32\textwidth]{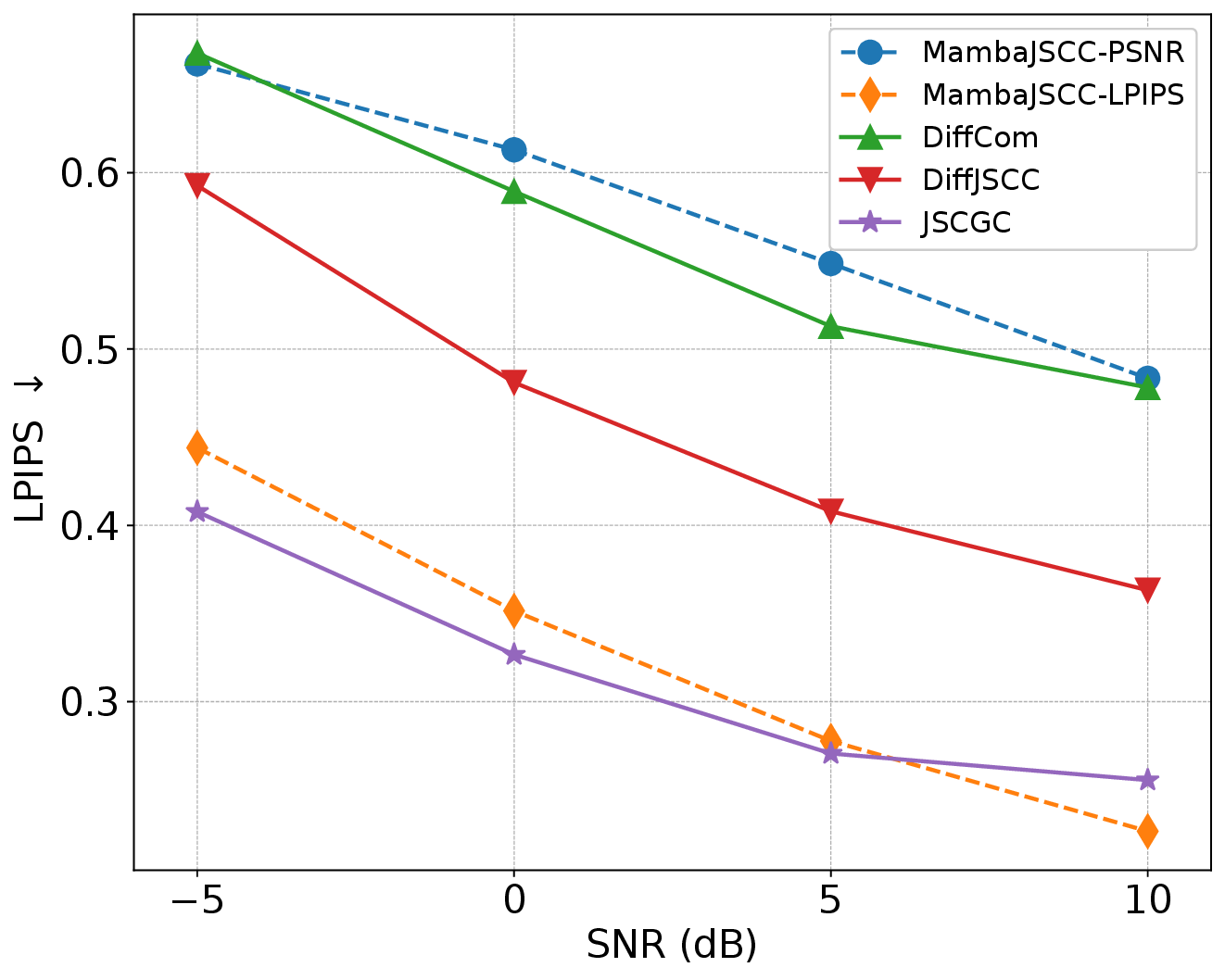}}
  \subfigure[]{\includegraphics[width=0.32\textwidth]{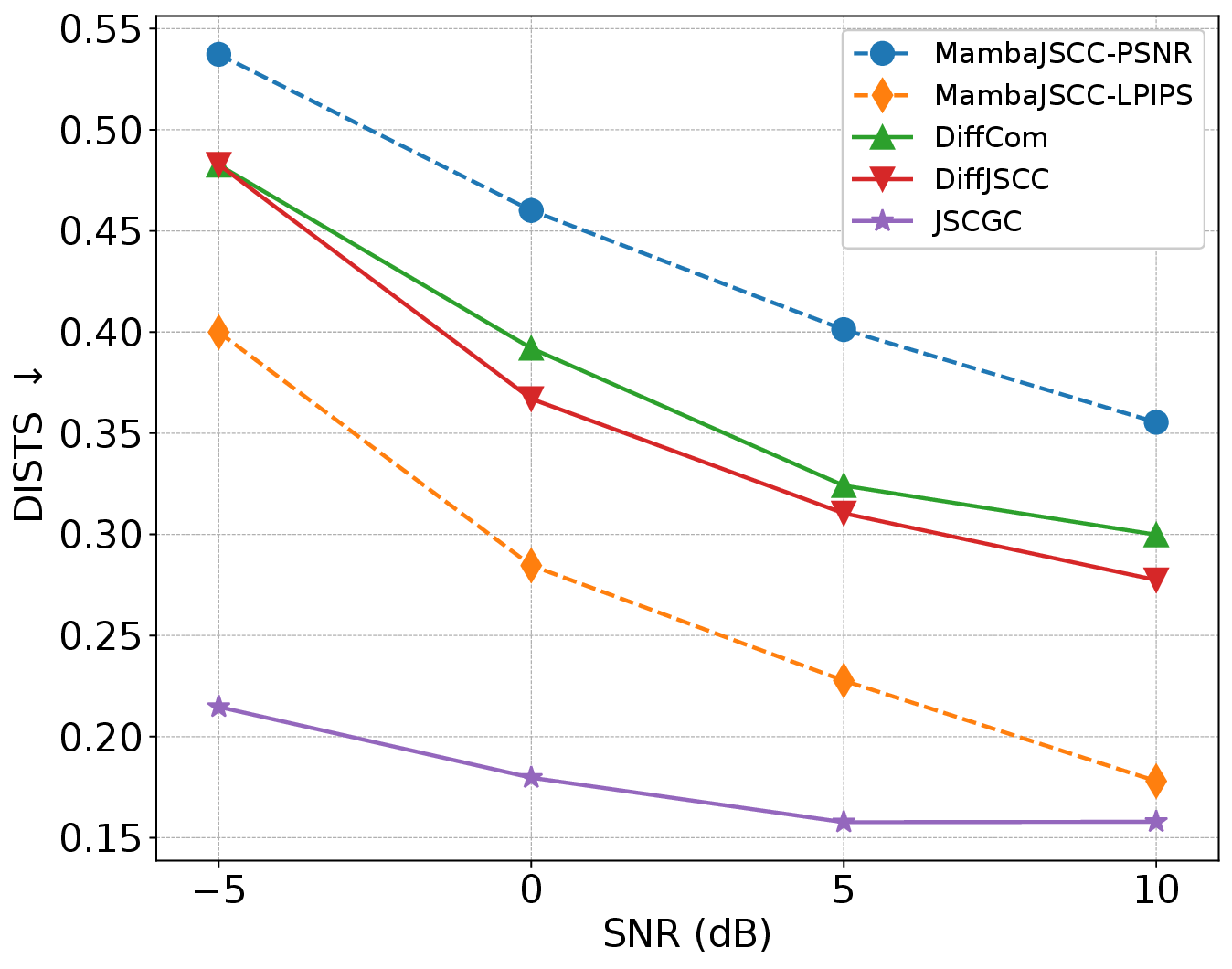}}
  \subfigure[]{\includegraphics[width=0.32\textwidth]{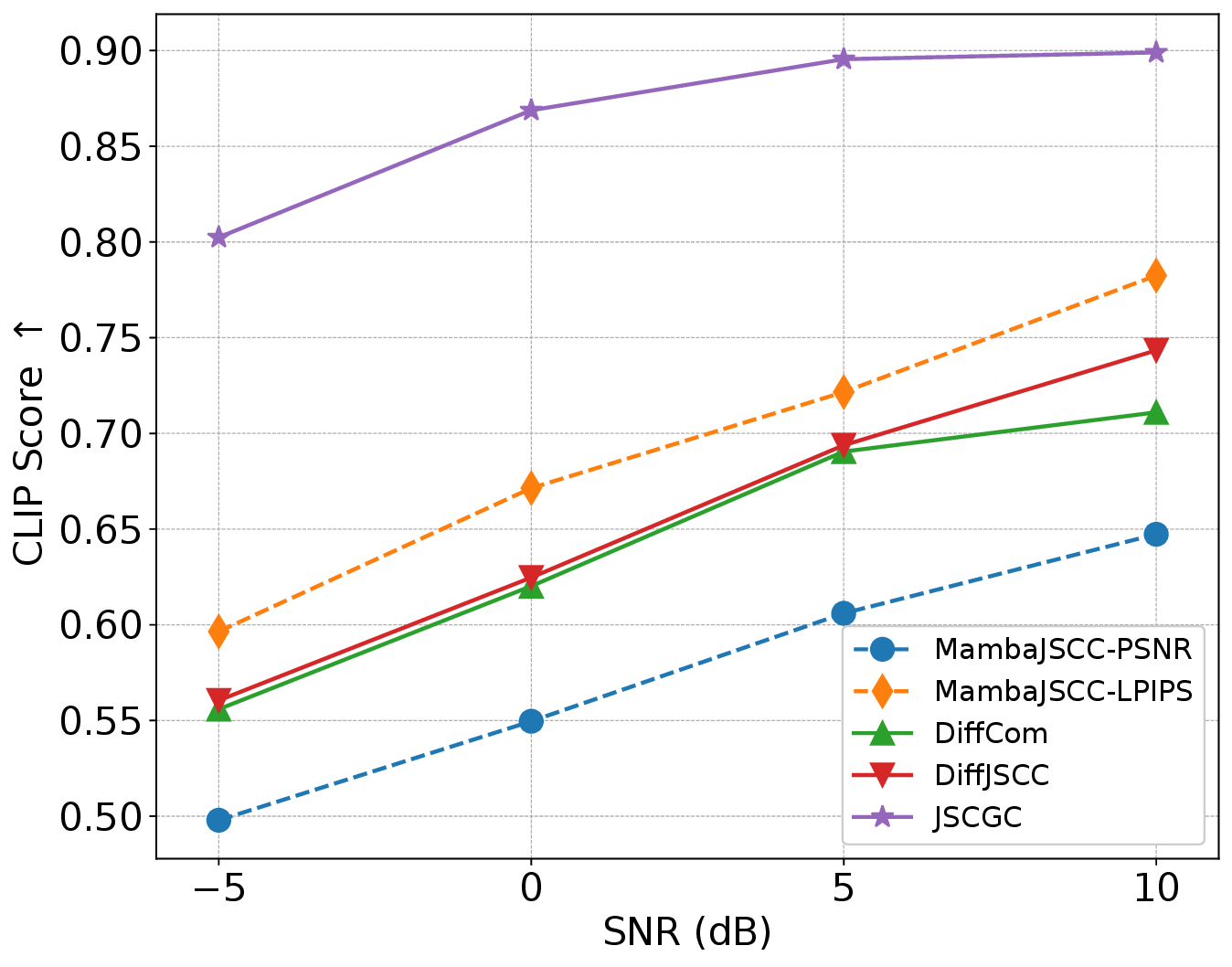}}
  \subfigure[]{\includegraphics[width=0.32\textwidth]{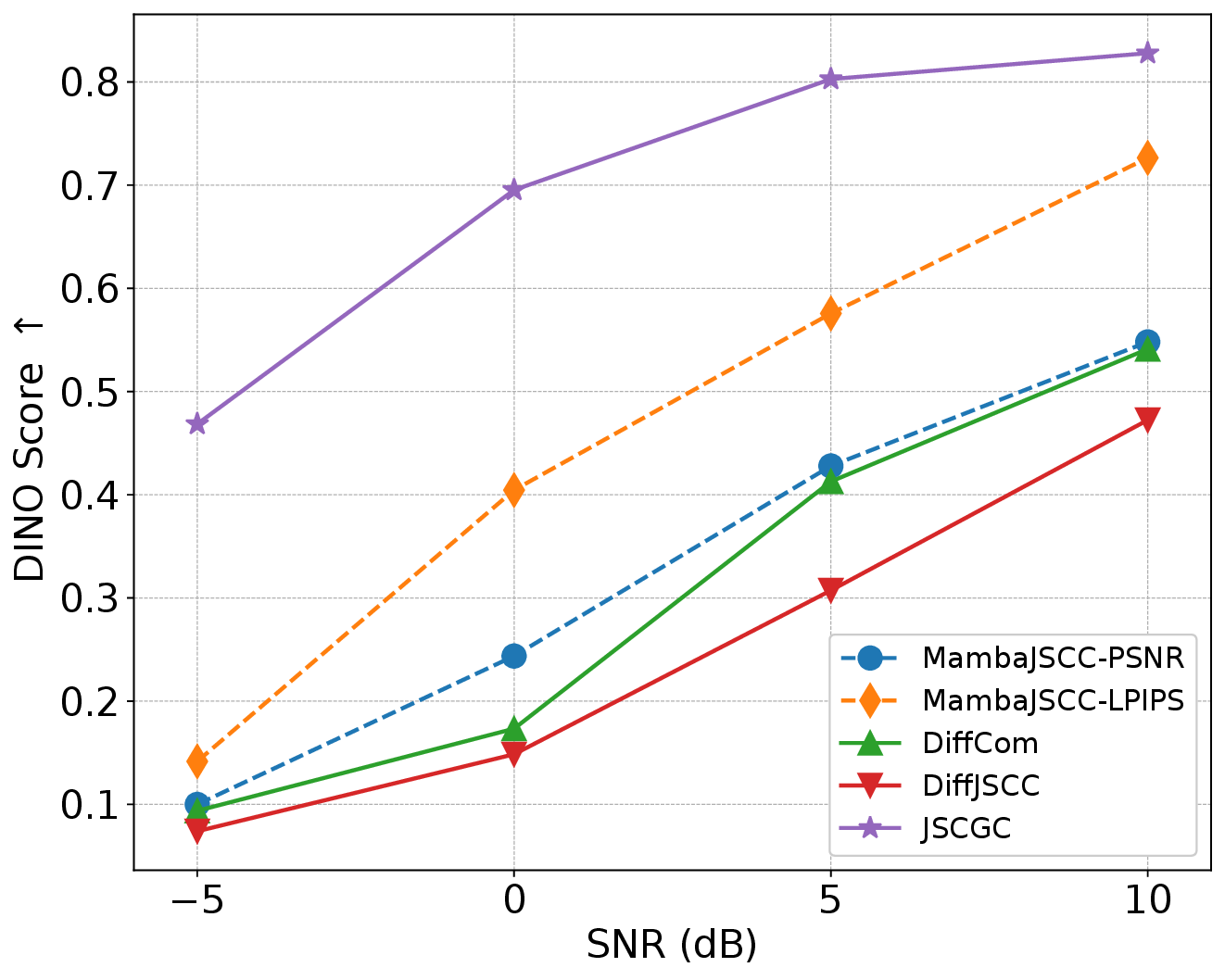}}
  \subfigure[]{\includegraphics[width=0.32\textwidth]{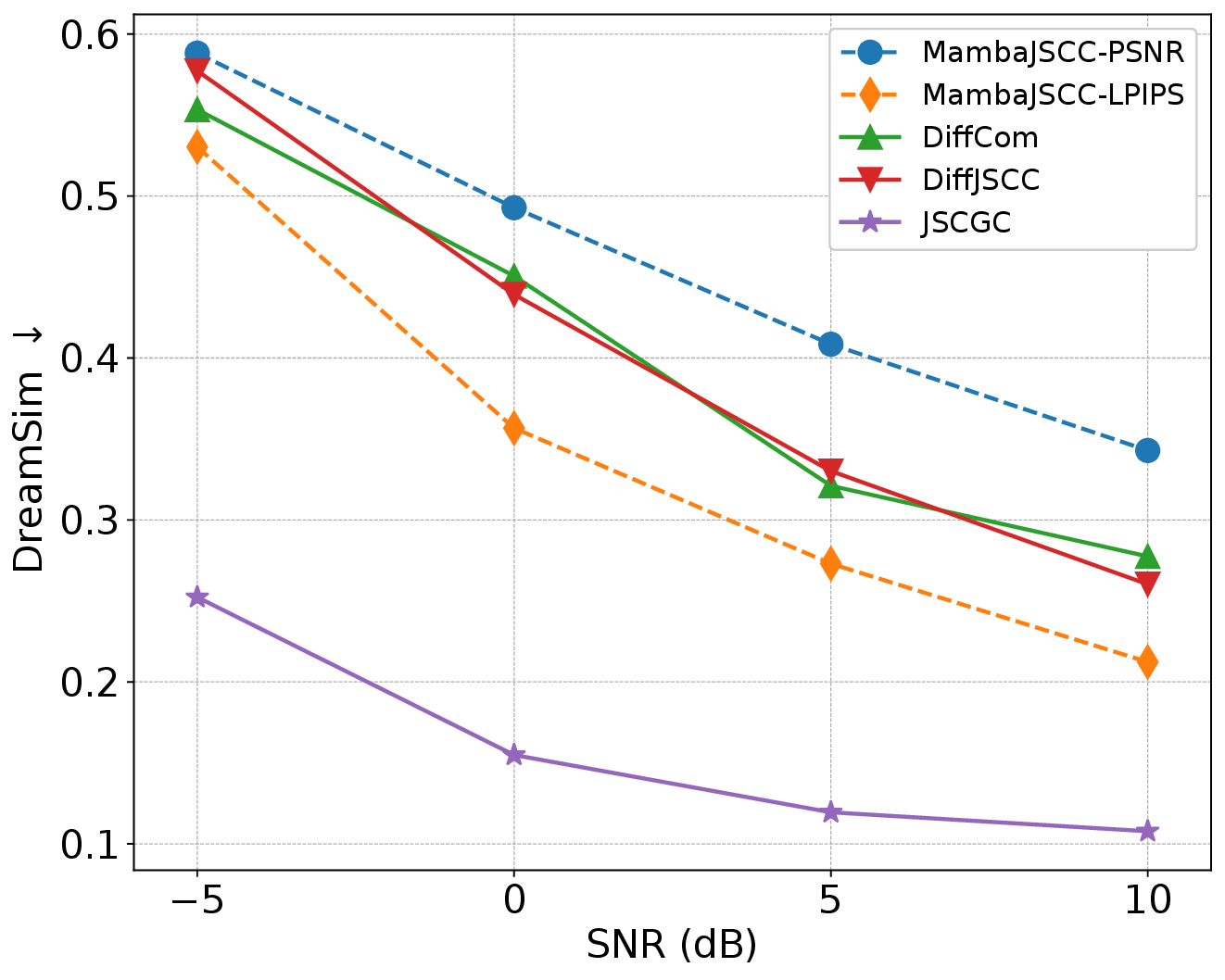}}
  \subfigure[]{\includegraphics[width=0.32\textwidth]{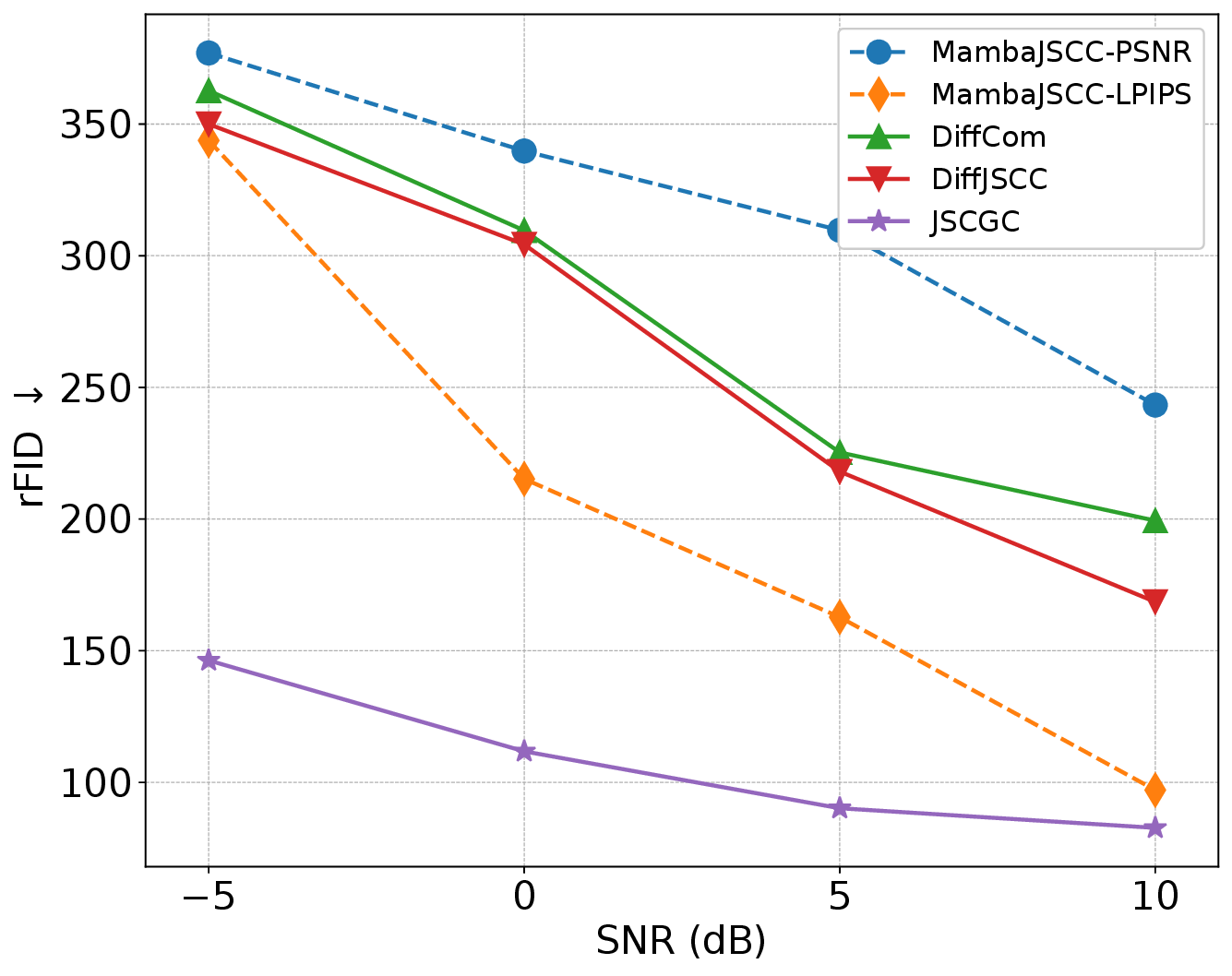}}
  \caption{The performance of different schemes versus SNR under the Rayleigh fading channel with CBR = $\frac{2}{768}$. (a) LPIPS. (b) DISTS. (c) CLIP Score. (d) DINO Score. (e) DreamSim. (f) rFID.} 
  \label{rayleigh_results}
  \vspace{-0.3cm}
\end{figure*}


\subsection{Experimental Setup}
We utilize a subset of the Open Images dataset consisting of 500k randomly sampled images for training. All images are randomly cropped to a resolution of $256 \times 256$. The Kodak dataset at its original resolution is used for evaluation.

We compare the proposed JSCGC scheme, implemented as shown in Fig. \ref{JSCGC_model}, with four representative baselines, including two reconstruction-oriented JSCC schemes and two diffusion-based generative schemes. The JSCC baselines are \textbf{MambaJSCC-PSNR} and \textbf{MambaJSCC-LPIPS} \cite{MambaJSCCWu}, trained with MSE and LPIPS losses, respectively. The generative baselines are \textbf{DiffCom} \cite{DiffCom}, where the original JSCC module is replaced by MambaJSCC and   the ``HiFi-DiffCom'' version is adopted, and \textbf{DiffJSCC} \cite{DiffJSCC}, which fine-tunes the Stable Diffusion for conditional image generation.

Experiments are conducted over both the additive white Gaussian noise (AWGN) channel and the Rayleigh fading channel. Following the setting in \cite{DiffJSCC}, the channel bandwidth ratio (CBR) is set to $\frac{n}{l} \in \left\{\frac{1}{768}, \frac{2}{768}, \frac{4}{768}, \frac{8}{768}\right\}$.

To comprehensively evaluate the proposed JSCGC scheme, we adopt a diverse set of metrics covering multiple levels of image fidelity and consistency. Specifically, we use peak signal-to-noise ratio (PSNR) as a pixel-level distortion metric; LPIPS and DISTS \cite{DISTS} as feature-level similarity metrics; CLIP Score \cite{CLIPscore} and DINO Score \cite{DINOv2} as semantic-level consistency metrics; DreamSim \cite{DreamSim} as a holistic visual similarity metric; and reconstruction FID (rFID) \cite{rFID} as a distribution-level similarity metric.
Among these metrics, DISTS prefers structural and texture consistency. CLIP Score measures the high-level semantic similarity and is computed using the ConvNext-based CLIP model. DINO Score also measures semantic consistency between image pairs, but is based on the DINOv2-ViT-L/14 model. DreamSim provides a holistic perceptual similarity measure by integrating CLIP, OpenCLIP, and DINOv2 representations. Unlike the above pairwise image similarity metrics, rFID evaluates distribution level consistency by measuring the similarity between the distribution of generated images and that of real reference images.

To accelerate convergence, the JSCGC is initialized using a pretrained Z-Image model and a pretrained MambaJSCC encoder. We first perform a joint training stage for $100$k iterations at an SNR of $10$ dB and a CBR of $\frac{60}{768}$. The resulting checkpoint is then used to initialize subsequent fine-tuning under each target channel condition for an additional $20$k iterations. All training is conducted on four NVIDIA RTX 4090 GPUs using DeepSpeed ZeRO-2 with a global batch size of $1$. 

\subsection{Results}

We first evaluate JSCGC over the AWGN channel with $CBR=\frac{2}{768}$. The results are shown in Fig.~\ref{awgn_results}, where the arrows indicate whether higher or lower values are preferred. We can observe that JSCGC consistently achieves the best performance across most metrics and SNRs, with particularly significant gains in the low SNR regime. At an SNR of $-5$ dB, JSCGC substantially outperforms all baselines. Compared with the strongest baseline, MambaJSCC-LPIPS, JSCGC improves the CLIP Score and DINO Score by 1.34$\times$ and 3.19$\times$, respectively, while reducing LPIPS, DISTS, DreamSim, and rFID by approximately 11\%, 42\%, 54\%, and 54\%. As the SNR increases, the performance gap gradually narrows; however, JSCGC still maintains the best overall performance. For example, at an SNR of $10$ dB, it continues to achieve higher semantic consistency and lower perceptual distance than DiffJSCC across all evaluation metrics. These results validate the advantage of the proposed JSCGC paradigm. Unlike DiffCom and DiffJSCC, which rely on a reconstruction-oriented decoder and optimize generation separately from encoding and transmission, JSCGC jointly optimizes source encoding, channel transmission, and generation. Consequently, the transmitted signal serves as a more informative control condition for the generator, leading to improved semantic consistency and perceptual quality, particularly under severe channel impairments.

It is worth noting that, although MambaJSCC-LPIPS achieves highly competitive quantitative performance and even surpasses JSCGC in LPIPS at an SNR of $10$ dB, the visual results in Figs. \ref{Visualizing1} reveal a different conclusion. Specifically, MambaJSCC-LPIPS often produces noticeable checkerboard artifacts, characterized by grid-like textures and grayish color bias, particularly under severe channel distortion. In contrast, such artifacts are largely absent in the generative approaches. This observation reflects the Goodhart phenomenon: once an evaluation metric becomes the optimization objective, the model may learn to exploit the metric rather than faithfully preserve natural image quality. Consequently, superior metric values do not necessarily correspond to better perceptual realism. By avoiding explicit optimization toward distortion or perceptual metrics, JSCGC achieves consistently strong performance across diverse metrics while maintaining high visual quality, indicating that its gains stem from improved generation fidelity rather than metric-specific overfitting.

Furthermore, Fig. \ref{Visualizing1} illustrates a fundamentally different degradation behavior. As the SNR decreases, all baseline methods suffer from increasingly severe blurring or artificial textures. In contrast, JSCGC preserves realistic visual appearance even under highly noisy channel conditions. The primary degradation manifests as reduced semantic consistency with the source image rather than perceptual quality loss. This observation experimentally validates the theoretical insights in Remark \ref{change}, demonstrating that JSCGC shifts communication errors from perceptual distortion to controllable semantic deviation.

In terms of PSNR, as shown in Fig. \ref{PSNR_awgn_C4}, JSCGC exhibits lower performance than reconstruction-oriented methods. This observation is consistent with the RDP theory, which characterizes the inherent trade-off between distortion minimization and perceptual quality. Since JSCGC is optimized toward perceptual fidelity rather than pixel-wise reconstruction accuracy, a reduction in PSNR is expected. Nevertheless, the results across perceptual, semantic, and distribution-level metrics consistently demonstrate that JSCGC generates images with high visual quality and strong semantic consistency.

To further evaluate the robustness of JSCGC, we conduct additional experiments under Rayleigh fading channels and different CBR settings. Fig. \ref{rayleigh_results} shows the performance versus SNR at a CBR of $2/768$. Compared with the AWGN channel scenario, JSCGC achieves even larger performance gains under Rayleigh fading channels, and these gains remain significant across the entire SNR range. For example, at an SNR of $5$ dB, JSCGC improves the CLIP Score and DINO Score from 0.6905 and 0.4128 to 0.8955 and 0.8028, respectively, while reducing DreamSim and rFID from 0.3211 and 225 to 0.1195 and 90.17, compared with DiffCom.

Fig. \ref{awgn_results_cbr} presents the performance versus CBR under the AWGN channel with an SNR of $5$ dB. JSCGC consistently outperforms the baseline methods across almost all CBR settings. The performance gain is particularly pronounced at low CBRs and gradually decreases as CBR increases. This trend is consistent with the RDP theory: as more channel resources become available, the discrepancy between distortion-oriented optimization and perceptual optimization diminishes, resulting in a smaller performance gap.

\begin{table}[t]
\centering
\caption{Computational complexity comparison.}
\label{tab:complexity}
\renewcommand{\arraystretch}{1.15}
\begin{tabularx}{\columnwidth}{lccc}
\toprule
Method & FLOPs (G) & Params (M) & GPU Memory (MB) \\
\midrule
MambaJSCC-PSNR & 245.72 & 11.69 & 1994  \\
DiffCom        & -- & 564.5 & 44440 \\
DiffJSCC       & 1867.06 & 5380.49 & 18626 \\
JSCGC          & 22753.68 & 7724.84 & 16426 \\
\bottomrule
\end{tabularx}
\end{table}

\begin{table}[t]
\centering
\caption{Comparison of sampling step, GPU time and performance.}
\label{tab:comparison}
\setlength{\tabcolsep}{0.1pt}
\renewcommand{\arraystretch}{1.3}

\begin{tabularx}{\columnwidth}{@{}>{\raggedright\arraybackslash}p{0.16\columnwidth}
                                  >{\centering\arraybackslash}X
                                   >{\centering\arraybackslash}X
                                  >{\centering\arraybackslash}X
                                   >{\centering\arraybackslash}X
                                  >{\centering\arraybackslash}X
                                   >{\centering\arraybackslash}X@{}}
\toprule
\multirow{2}{*}{Method}
& \multirow{2}{*}{\makecell{Sampling\\Step}}
& \multirow{2}{*}{\makecell{GPU\\Time (s)}}
& \multicolumn{2}{c@{}}{Performance (0 dB)} 
& \multicolumn{2}{c@{}}{Performance (10 dB)} \\
\cmidrule(l){4-5} \cmidrule(l){6-7}
& & &\makebox[\linewidth][c]{\hspace{5pt}DreamSim} & rFID  &\makebox[\linewidth][c]{\hspace{5pt}DreamSim} & rFID \\

\midrule
DiffCom  & 252 & 419.54 & 0.4986 & 320.88 &  --    &  --    \\
DiffCom  & 235 & 392.82 &   --   &   --   & 0.2122 & 125.09 \\
DiffJSCC & 50  & 9.08  & 0.2851  & 195.11 & 0.1416 & 91.57  \\
\midrule
JSCGC-10   & 10  & 5.09 & 0.1707 & 111.29 & 0.1283 & 94.19 \\
JSCGC-20   & 20  & 8.62 & 0.1490 & 98.23  & 0.1045 & 83.76 \\
JSCGC-30   & 30  & 12.94 & 0.1420 & 99.07 & 0.0979 & 84.94 \\
\textbf{JSCGC-50}   & \textbf{50}  & \textbf{21.62} & \textbf{0.1388} & \textbf{96.62} & \textbf{0.0955} & \textbf{81.01} \\
JSCGC-100  & 100 & 43.42 & 0.1390 & 95.21 & 0.0957 & 81.68\\
\bottomrule
\end{tabularx}
\end{table}

In addition to transmission performance, we evaluate the computational complexity and deployment cost of different schemes in Table \ref{tab:complexity}. FLOPs are measured using PyTorch Profiler. Compared with reconstruction-based methods, generative approaches generally require substantially larger models and higher computational cost. For example, MambaJSCC-PSNR contains only 11.69M parameters and requires 245.72G FLOPs and 1994 MB of GPU memory for a $512\times768$ image. By contrast, JSCGC adopts the large-scale Z-Image generator and therefore has the largest parameter size (7724.84M). Nevertheless, thanks to efficient inference optimization, JSCGC requires the lowest GPU memory among all generative methods, consuming only 16426 MB during inference.
\begin{figure*}[t]
  \centering
  \subfigure[]{\includegraphics[width=0.32\textwidth]{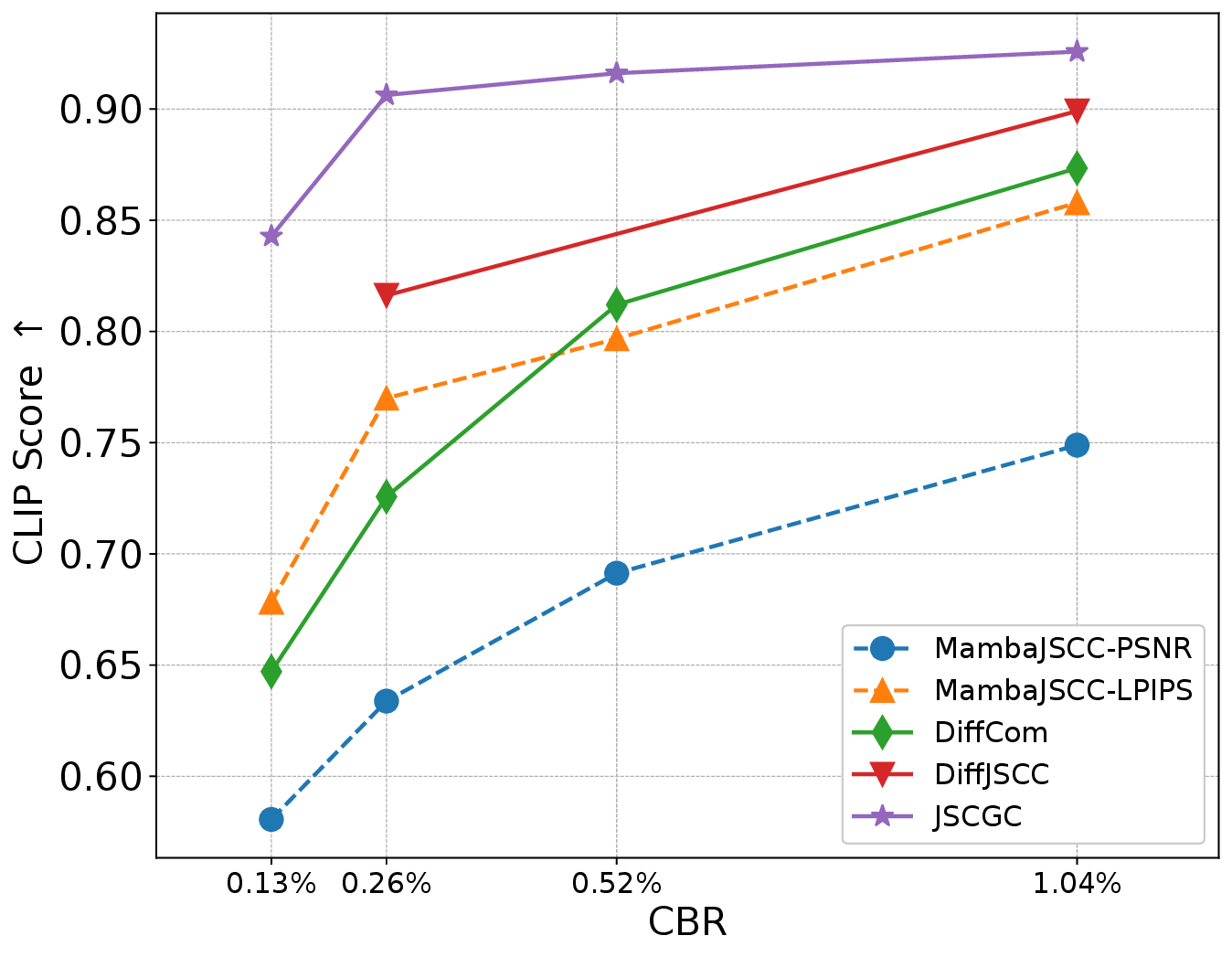}}
  \subfigure[]{\includegraphics[width=0.32\textwidth]{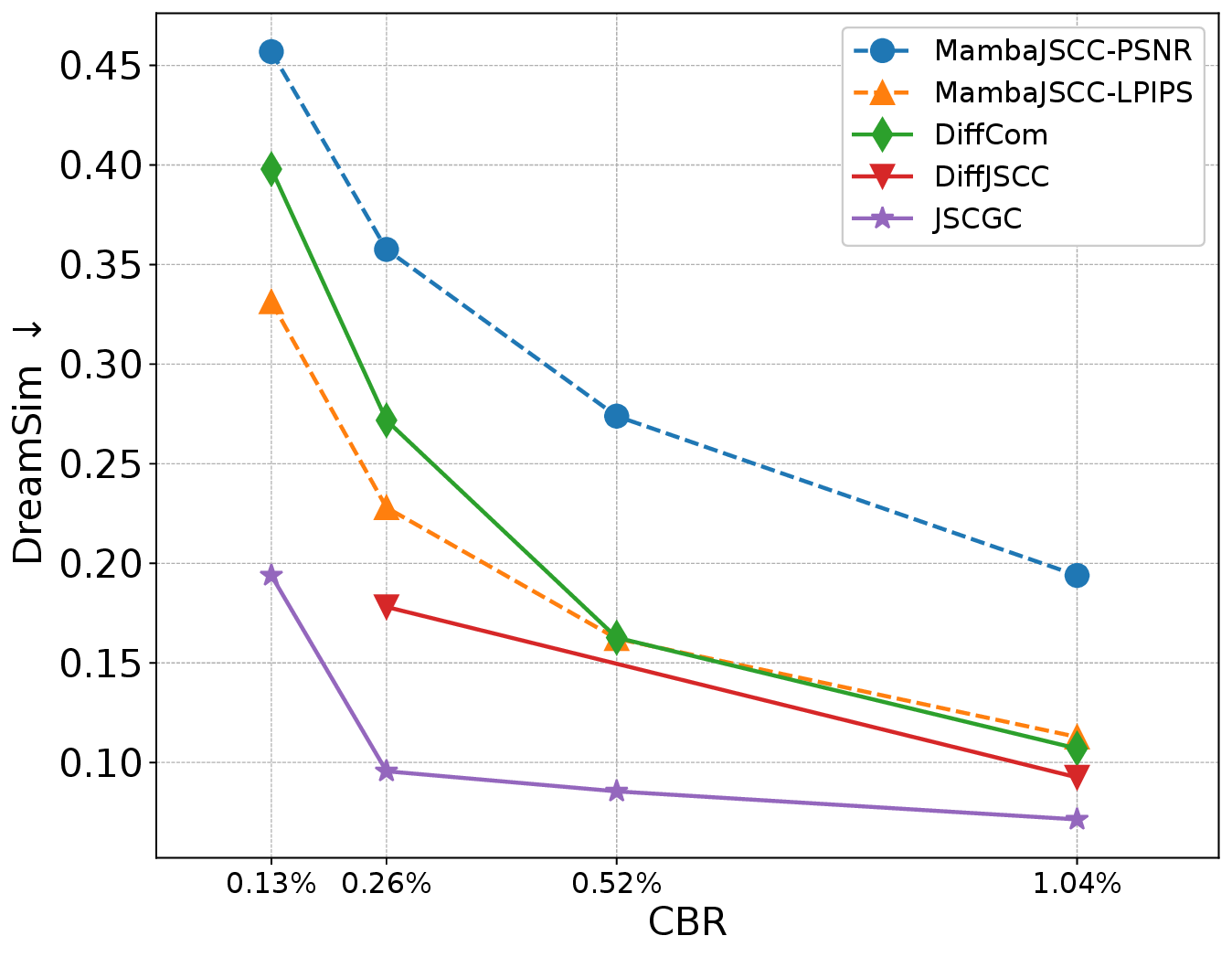}}
  \subfigure[]{\includegraphics[width=0.32\textwidth]{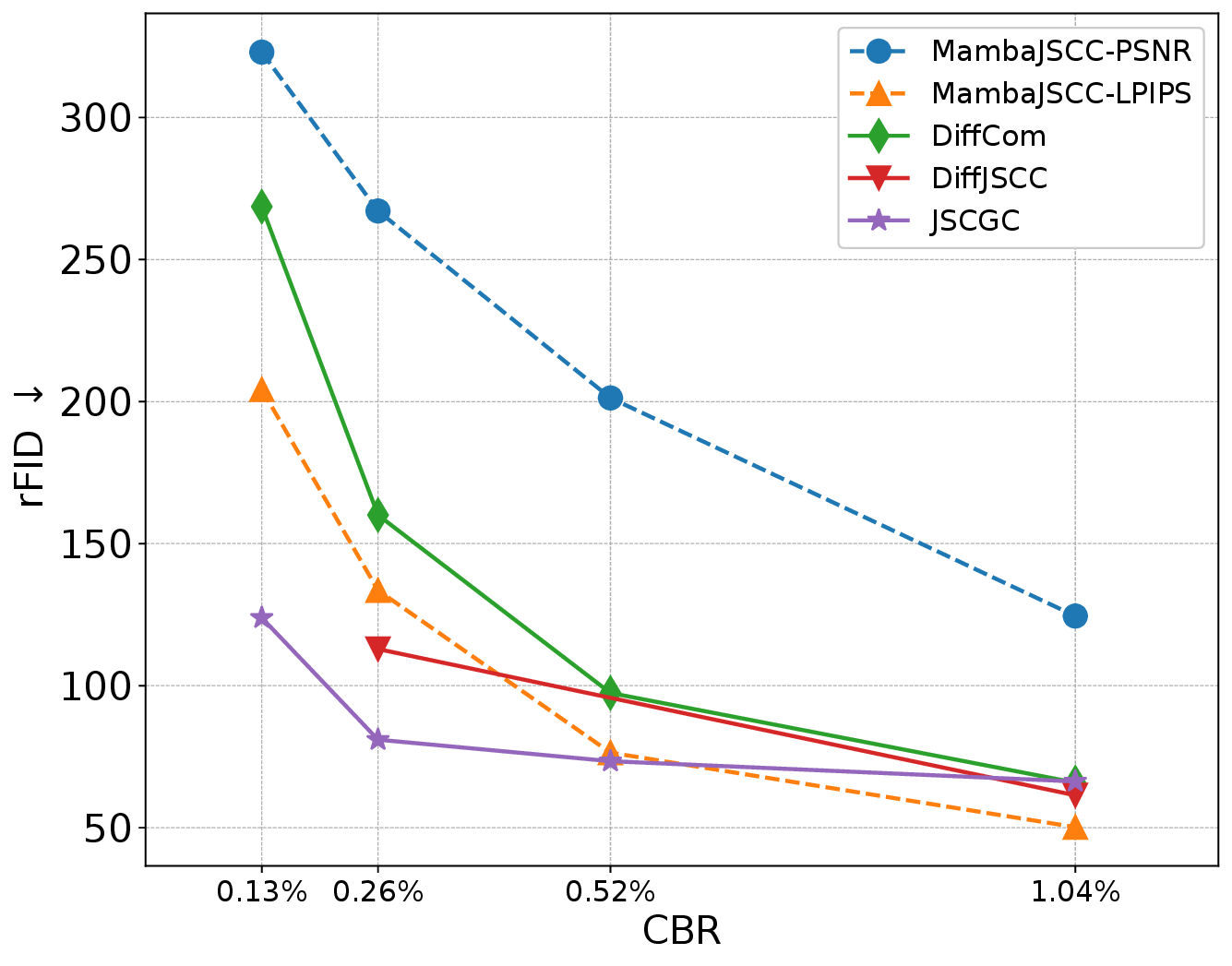}}
  \caption{The performance of different schemes versus CBR under the AWGN channel with SNR = 5 dB.  (a) CLIP Score. (b) DreamSim. (c) rFID.} 
  \label{awgn_results_cbr}
  \vspace{-0.3cm}
\end{figure*}

Since different methods employ different numbers of sampling steps, and DiffCom additionally performs gradient-based optimization during generation, FLOPs alone cannot fully reflect practical deployment cost. Therefore, Table \ref{tab:comparison} further reports the actual inference latency and corresponding performance at SNRs of 0 dB and 10 dB with a CBR of $2/768$. DreamSim and rFID are selected as representative evaluation metrics. Here, JSCGC-$N$ denotes the proposed method with $N$ sampling steps. The results show that DiffCom incurs the highest latency, requiring more than 18 times the inference time of JSCGC-50 while still achieving substantially worse performance. DiffJSCC reduces latency to 9.08 seconds, approximately half that of JSCGC-50. However, JSCGC can reduce the number of sampling steps to 20, lowering the latency to 8.62 seconds while still outperforming DiffJSCC under both SNR conditions. We further examine the performance-latency trade-off of JSCGC. Increasing the number of sampling steps from 10 to 50 consistently improves performance at the cost of higher latency. In contrast, increasing the number of steps from 50 to 100 yields only marginal gains. Therefore, 50 sampling steps are adopted as the default setting throughout this paper, providing a favorable balance between performance and efficiency.

\section{Conclusion}\label{sec:conclusion}
In this paper, we proposed JSCGC, a joint source-channel generative communication framework that shifts wireless communication from deterministic reconstruction to conditional generation. By maximizing mutual information under perceptual constraints, JSCGC enables generation-oriented encoding and transmission without relying on explicit distortion metrics or reconstruction decoders. We developed corresponding training and sampling algorithms and analyzed their effectiveness theoretically. Experimental results demonstrated that JSCGC consistently outperforms representative reconstruction-based and generative communication schemes in feature-based distortion, semantic consistency, and distributional fidelity, while exhibiting a distinct error behavior characterized by semantic inconsistency rather than perceptual distortion.
\appendices
\section{Proof of Proposition~\ref{prop:variational_info}}\label{app:training_loss}
Reviewing the communication process in (\ref{CMC}), we aim to generate $\hat{\mathbf{x}}$ under the control of $\hat{\mathbf{y}}$, where $\hat{\mathbf{y}}$ is the received signal and serves as the input of the generator, eliminating the need for parameterized generation. Therefore, we define the joint distributions with the true distribution regarding $\hat{\mathbf{y}}$ that
  $p(\mathbf{x}, \hat{\mathbf{y}}) = p(\hat{\mathbf{y}}) p(\mathbf{x}|\hat{\mathbf{y}}),$ and 
  $p_\theta(\mathbf{x}, \hat{\mathbf{y}}) = p(\hat{\mathbf{y}}) q_\theta(\mathbf{x}|\hat{\mathbf{y}})$.
The KL divergence can be expressed as 
  \begin{align}
  &D_{KL}(p(\mathbf{x},\hat{\mathbf{y}})||q_\theta(\mathbf{x},\hat{\mathbf{y}})) = D_{KL}(p(\hat{\mathbf{y}})||p(\hat{\mathbf{y}})) \nonumber\\ &+ \int p(\hat{\mathbf{y}}) \int p(\mathbf{x}|\hat{\mathbf{y}}) \log \frac{p(\mathbf{x}|\hat{\mathbf{y}})}{q_\theta(\mathbf{x}|\hat{\mathbf{y}})} d\mathbf{x} d\hat{\mathbf{y}},\\
  &D_{KL}(p(\mathbf{x},\hat{\mathbf{y}})||q_\theta(\mathbf{x},\hat{\mathbf{y}})) \geq D_{KL}(p(\mathbf{x})||q_\theta(\mathbf{x})) 
  \end{align}
Therefore, in our framework, we have
  \begin{align}\label{eq:uncondKLwithcondKL}
    D_{KL}(p(\mathbf{x})||q_\theta(\mathbf{x})) &\leq \int p(\hat{\mathbf{y}}) \int p(\mathbf{x}|\hat{\mathbf{y}}) \log \frac{p(\mathbf{x}|\hat{\mathbf{y}})}{q_\theta(\mathbf{x}|\hat{\mathbf{y}})} d\mathbf{x} d\hat{\mathbf{y}}\nonumber\\
    &= \iint p(\mathbf{x}, \hat{\mathbf{y}}) \log \frac{p(\mathbf{x}|\hat{\mathbf{y}})}{q_\theta(\mathbf{x}|\hat{\mathbf{y}})} d\mathbf{x} d\hat{\mathbf{y}}.
  \end{align}
  The optimization objective can be derived as 
  \begin{align}
    &H(\mathbf{X} | \hat{\mathbf{Y}})+D_{KL}(p(\mathbf{x})||q_\theta(\mathbf{x}))\nonumber\\
    &\leq \iint (-p(\mathbf{x},\hat{\mathbf{y}})[ \log \frac{p(\mathbf{x}|\hat{\mathbf{y}}))}{p(\mathbf{x}|\hat{\mathbf{y}}))} + \log q_\theta(\mathbf{x}|\hat{\mathbf{y}})] d\mathbf{x} d\hat{\mathbf{y}}\nonumber\\
    &=\mathbb{E}_{p(\mathbf{x},\hat{\mathbf{y}})}\left[-\log q_\theta(\mathbf{x}|\hat{\mathbf{y}})\right] 
  \end{align} 
Furthermore, we construct a sequence of random variables $\{\mathbf{x}_t\}^T_{t=0}$ as 
$\mathbf{x}_t=(1-\frac{t}{T})\mathbf{x}_0+\frac{t}{T} \mathbf{\epsilon}$, where $\mathbf{x}_0=\mathbf{x}$.

The evidence upper bound with respect to $\{\mathbf{x}_t\}^T_{t=0}$ is:
\begin{equation}
\label{ELBO}
\begin{aligned}[c]
&\mathbb{E}_{p(\mathbf{x}_0,\hat{\mathbf{y}})}\left[-\log q_\theta(\mathbf{x}_0|\hat{\mathbf{y}})\right] \leq \mathbb{E}_{p(\mathbf{x}_{0:T},\hat{\mathbf{y}})}\left[-\log \frac{q_\theta(\mathbf{x}_{0:T}|\hat{\mathbf{y}})}{p(\mathbf{x}_{1:T}|\mathbf{x}_0, \hat{\mathbf{y}})}\right]\\
&=\mathbb{E}_{p(\mathbf{x}_{0:1},\hat{\mathbf{y}})} (
-\underbrace{\log q_\theta(\mathbf{x}_0 | \mathbf{x}_1, \hat{\mathbf{y}})}_{\text{Reconstruction Term}} )\\
&+ \mathbb{E}_{p(\mathbf{x}_{0})} \Big(\underbrace{D_{\text{KL}}(p(\mathbf{x}_T|\mathbf{x}_0) \,\|\, p(\mathbf{x}_T))}_{\text{Prior Matching Term}}\Big)\\
& + \sum_{t=2}^{T} \underbrace{\mathbb{E}_{p(\mathbf{x}_{0:t},\hat{\mathbf{y}})} \Big(D_{\text{KL}}(p(\mathbf{x}_{t-1}|\mathbf{x}_t, \mathbf{x}_0, \hat{\mathbf{y}}) \,\|\, q_\theta(\mathbf{x}_{t-1}|\mathbf{x}_t, \hat{\mathbf{y}}))\Big)}_{\text{Posterior Term}}.
\end{aligned}
\end{equation}
Here, we only need to optimize the posterior term. In this term, the true posterior is tractable as  
\begin{align}
p(\mathbf{x}_{t-1}|\mathbf{x}_t, \mathbf{x}_0, \hat{\mathbf{y}})= \mathcal{N}\left(\mathbf{x}_{t-1};  \boldsymbol{\mu}(\mathbf{x}_t, t, \mathbf{x}_0), \sigma^2 \mathbf{I}\right).
\end{align}
where  $\sigma^2$ is $\frac{(t-1)^2 (2Tt - 2t^2 + 2t - T)}{T t^2 (T-t+1)^2},$ and $\boldsymbol{\mu}(\mathbf{x}_t, t, \mathbf{x}_0)$ is
\[
\frac{(t-1)^2 (T-t)}{t^2 (T-t+1)}\mathbf{x}_t + \left(\frac{T-t+1}{T} - \frac{(t-1)^2 (T-t)^2}{t^2 (T-t+1) T}\right)\mathbf{x}_0.
\]

Therefore, we also formulate the parameterized distribution as a Gaussian distribution with mean $\boldsymbol{\mu}_\theta(\mathbf{x}_t, t, \hat{\mathbf{y}})$ and variance $\sigma^2 \mathbf{I}$. Therefore, we can compute the posterior term as:
\begin{equation}\label{KL_divergence}
\begin{aligned}
&D_{\text{KL}}\left( \mathcal{N}(\boldsymbol{\mu}(\mathbf{x}_t, \mathbf{x}_0, t), \sigma^2 \mathbf{I}) \,\|\, \mathcal{N}(\boldsymbol{\mu}_\theta(\mathbf{x}_t, t, \hat{\mathbf{y}}), \sigma^2 \mathbf{I}) \right) \\
&= \frac{1}{2\sigma^2} \| \boldsymbol{\mu}(\mathbf{x}_t, \mathbf{x}_0, t) - \boldsymbol{\mu}_\theta(\mathbf{x}_t, t, \hat{\mathbf{y}}) \|^2. \hspace{-0.5cm}
\end{aligned}
\end{equation}

\addtolength{\topmargin}{0.02in}
To facilitate the training efficiency, we reparameterize the posterior mean to predict the velocity field $\boldsymbol{\nu} = \frac{T}{t}(\mathbf{x}_t-\mathbf{x}_0)=\mathbf{\epsilon} - \mathbf{x}_0$. Therefore, $\boldsymbol{\mu}(\mathbf{x}_t, t, \mathbf{x}_0)$ is
$\frac{1}{t}\mathbf{x}_t + \left[ - \frac{2t-1}{t(T-t+1)} \right](\mathbf{\epsilon} - \mathbf{x}_0)$.
We further reparameterize $\mathbf{\mu}_\theta(\mathbf{x}_t, t, \hat{\mathbf{y}}) $
and substitute these into (\ref{KL_divergence}), resulting in the final tractable optimization objective as:
\begin{equation}\label{training_loss}
\mathcal{L}(\theta,\phi)= \mathbb{E}_{t, \mathbf{x}_0, \hat{\mathbf{y}}, \mathbf{\epsilon}} \left[ \| (\mathbf{\epsilon} - \mathbf{x}_0) - \boldsymbol{\nu}_\theta(\mathbf{x}_t, t, \hat{\mathbf{y}}) \|^2 \right].
\end{equation}

\section{Proof of Proposition~\ref{prop:optimal_solution}}\label{app:optimal_solution}
Expanding the loss function into its integral form yields:
\[
    \iiiint\limits_{t, \mathbf{x}_0, \hat{\mathbf{y}}, \mathbf{\epsilon}} \left\| (\mathbf{\epsilon} - \mathbf{x}_0) - \boldsymbol{\nu}_\theta(\mathbf{x}_t, t, \hat{\mathbf{y}}) \right\|^2
    p(t, \mathbf{x}_0, \hat{\mathbf{y}}, \mathbf{\epsilon}) \, d\mathbf{\epsilon} \, d\hat{\mathbf{y}} \, d\mathbf{x}_0 \, dt
\]
Here, $(\mathbf{x}_t, t, \hat{\mathbf{y}})$ is input of $\boldsymbol{\nu}_\theta(\cdot)$ and $\mathbf{x}_t$ is a deterministic function of $(\mathbf{x}_0, \mathbf{\epsilon}, t)$. Therefore, $p(t, \mathbf{x}_0, \hat{\mathbf{y}}, \mathbf{\epsilon}) \, d\mathbf{\epsilon} \, d\hat{\mathbf{y}} \, d\mathbf{x}_0 \, dt$  is
\begin{align}
p(\mathbf{x}_0, \mathbf{\epsilon} \mid \mathbf{x}_t, t, \hat{\mathbf{y}}) \, d\mathbf{\epsilon} \, d\mathbf{x}_0 \cdot \ p(\mathbf{x}_t, t, \hat{\mathbf{y}}) \, d\hat{\mathbf{y}} \, dt \, d\mathbf{x}_t
\end{align}
With the decomposition, the integral is rewritten as
\[
  \iiint\limits_{\mathbf{x}_t, t, \hat{\mathbf{y}}} \mathbb{E} \left[\left\| (\mathbf{\epsilon} - \mathbf{x}_0) - \boldsymbol{\nu}_\theta(\cdot) \right\|^2 \Big| \mathbf{x}_t, t, \hat{\mathbf{y}} \right]  p(\mathbf{x}_t, t, \hat{\mathbf{y}}) \, d\hat{\mathbf{y}} \, dt \, d\mathbf{x}_t
\]
 Noticing that the optimization variable $\theta,\phi$ only determine $\boldsymbol{\nu}_\theta(\cdot)$, optimizing the parameters is equivalent to optimizing $\boldsymbol{\nu}_\theta(\cdot)$.  Given the non-negativity of the integrand, minimizing the overall integral is equivalent to minimizing the integrand almost everywhere with respect to $(\mathbf{x}_t, t, \hat{\mathbf{y}})$. Therefore, let 
 \begin{equation}
  J(\boldsymbol{\nu}_\theta(\cdot))=\mathbb{E} \left[\left\| (\mathbf{\epsilon} - \mathbf{x}_0) - \boldsymbol{\nu}_\theta(\mathbf{x}_t, t, \hat{\mathbf{y}}) \right\|^2 \Big| \mathbf{x}_t, t, \hat{\mathbf{y}} \right].
 \end{equation}
 Expanding $J(\boldsymbol{\nu}_\theta(\cdot))$ with $\boldsymbol{\nu}= \mathbf{\epsilon}-\mathbf{x}_0$, we have:
\begin{align}
&J(\cdot)=\int \left( \|\boldsymbol{\nu}\|^2 - 2 \boldsymbol{\nu}_\theta(\cdot)^\top \boldsymbol{\nu} + \|\boldsymbol{\nu}_\theta(\cdot)\|^2 \right) p(\boldsymbol{\nu} \mid \mathbf{x}_t, t, \hat{\mathbf{y}}) \, d\boldsymbol{\nu} \nonumber\\
&= \underbrace{\int \|\boldsymbol{\nu}\|^2 \, p(\boldsymbol{\nu} \mid \mathbf{x}_t, t, \hat{\mathbf{y}}) \, d\boldsymbol{\nu}}_{\text{Independent of } \boldsymbol{\nu}_\theta(\cdot)} 
- 2 \boldsymbol{\nu}_\theta(\cdot)^\top \underbrace{\int \boldsymbol{\nu} \, p(\boldsymbol{\nu} \mid \mathbf{x}_t, t, \hat{\mathbf{y}}) \, d\boldsymbol{\nu}}_{\mathbb{E}[\boldsymbol{\nu} \mid \mathbf{x}_t, t, \hat{\mathbf{y}}]} \nonumber\\
&+ \|\boldsymbol{\nu}_\theta(\cdot)\|^2 \underbrace{\int p(\boldsymbol{\nu} \mid \mathbf{x}_t, t, \hat{\mathbf{y}}) \, d\boldsymbol{\nu}}_{=1}\nonumber\\
&= \|\mathbb{E}[\boldsymbol{\nu} \mid \mathbf{x}_t, t, \hat{\mathbf{y}}]-\boldsymbol{\nu}_\theta(\cdot)\|^2 + C, 
\end{align}
where $C=\int \|\boldsymbol{\nu}\|^2 \, p(\boldsymbol{\nu} \mid \mathbf{x}_t, t, \hat{\mathbf{y}}) \, d\boldsymbol{\nu}-\|\mathbb{E}[\boldsymbol{\nu} \mid \mathbf{x}_t, t, \hat{\mathbf{y}}]\|^2$ is constant independent of $\boldsymbol{\nu}_\theta(\cdot)$.

Obviously, if and only if $\boldsymbol{\nu}^*(\cdot)=\mathbb{E}[\boldsymbol{\nu} \mid \mathbf{x}_t, t, \hat{\mathbf{y}}]$, $J(\cdot)$ can be minimized. Therefore, the optimal solution is 
\begin{align}\label{eq:optimal_solution}
\boldsymbol{\nu}^*(\mathbf{x}_t, t, \hat{\mathbf{y}}) = \mathbb{E}[\mathbf{\epsilon} - \mathbf{x}_0|\mathbf{x}_t, t, \hat{\mathbf{y}}]
\end{align}
Furthermore, considering Tweedie's Formula 
\begin{align}
\nabla_{\mathbf{x}_t} \log p_t(\mathbf{x}_t|\hat{\mathbf{y}}) &= \mathbb{E}_{\mathbf{x}_0|\mathbf{x}_t,\hat{\mathbf{y}}} \left[ -\frac{T\mathbf{\epsilon}}{t}\Big|\mathbf{x}_t,\hat{\mathbf{y}} \right]\nonumber\\
&= - \frac{T-t}{t} \mathbb{E}[\mathbf{\epsilon} - \mathbf{x}_0 | \mathbf{x}_t, \hat{\mathbf{y}}] - \frac{T}{t}\mathbf{x}_t
\end{align}
 Finally, take (\ref{eq:optimal_solution}) into the above equation, we can derive 
\begin{equation}
  \boldsymbol{\nu}^*(\mathbf{x}_t, t, \hat{\mathbf{y}}) =  - \frac{t\nabla_{\mathbf{x}_t} \log p_t(\mathbf{x}_t|\hat{\mathbf{y}})+T\mathbf{x_t}}{T-t} 
\end{equation}

\section{Proof of Proposition \ref{SDEKLbound}}\label{app:SDEKLbound}
Let \(\Omega:=C([0,T];\mathbb R^d)\) be the path space equipped with its Borel \(\sigma\)-algebra \(\mathcal B(\Omega)\) induced by the supremum norm. The true reverse diffusion process in (\ref{eq:true_sde}) and the learned reverse diffusion process in (\ref{eq:learned_sde}) induce two probability measures on \((\Omega,\mathcal B(\Omega))\), denoted by $\mathbb{P}$ and $\mathbb{Q}$, respectively. Their path-space KL divergence is given by
$D_{\text{KL}}(\mathbb{P}\|\mathbb{Q}):=\mathbb E_{\mathbb{P}}\left[\log\frac{d\mathbb{P}}{d\mathbb{Q}}\right].$

Since they have the same marginal distributions, according to the chain rule of KL divergences, we have
\begin{align}
D_{\text{KL}}(\mathbb{P} \,\|\, \mathbb{Q}) = \mathbb{E}_{\mathbf{x}_T \sim \pi(\cdot)} \left[ D_{\text{KL}}\big(\mathbb{P}(\cdot \mid \mathbf{x}_T ) \,\|\, \mathbb{Q}(\cdot \mid \mathbf{x}_{T})\big) \right] \label{eq:chain_rule_kl}
\end{align}

Since the conditional path measures $\mathbb{P}(\cdot \mid \mathbf{x}_T)$ and $\mathbb{Q}(\cdot \mid \mathbf{x}_T)$ share identical initial conditions and the same diffusion coefficient $g(t)$, we can invoke the Girsanov theorem to explicitly compute their divergence.
By utilizing the martingale property of the It\^o integral under the true measure $\mathbb{P}$ ,which evaluates its expectation to zero, the conditional KL divergence simplifies directly to the integrated squared difference of the drifts:
\begin{align}
&D_{\text{KL}}\big(\mathbb{P}(\cdot \mid \mathbf{x}_T) \,\|\, \mathbb{Q}(\cdot \mid \mathbf{x}_T)\big) \nonumber\\
&= \frac{2}{T^2} \mathbb{E}_{\mathbb{P}(\cdot \mid \mathbf{x}_T)} \left[ \int_0^T \frac{1}{g(t)^2} \left\|\boldsymbol{\nu}^*(\mathbf{x}_t, t, \hat{\mathbf{y}}) - \boldsymbol{\nu}_\theta(\mathbf{x}_t, t, \hat{\mathbf{y}}) \right\|^2 dt \right] \nonumber \\
&= \int_0^T \frac{2}{g(t)^2 T^2} \mathbb{E}_{\mathbb{P}(\cdot \mid \mathbf{x}_T)} \left[ \big\| \boldsymbol{\nu}^*(\mathbf{x}_t, t, \hat{\mathbf{y}}) - \boldsymbol{\nu}_\theta(\mathbf{x}_t, t, \hat{\mathbf{y}}) \big\|^2 \right] dt \label{eq:conditional_girsanov}
\end{align}
According to the law of total expectation, nesting the outer expectation over the initial state $\mathbf{x}_T \sim p_T$, we have: 
\begin{align}
&D_{\text{KL}}(\mathbb{P} \,\|\, \mathbb{Q}) = \mathbb{E}_{\mathbf{x}_T \sim \pi(\cdot)} \left[ D_{\text{KL}}\big(\mathbb{P}(\cdot \mid \mathbf{x}_T ) \,\|\, \mathbb{Q}(\cdot \mid \mathbf{x}_{T})\big) \right]= \nonumber\\
&\int_0^T \frac{2}{g(t)^2 T^2} \mathbb{E}_{p_t(\mathbf{x}_t | \hat{\mathbf{y}})} \left[ \big\| \boldsymbol{\nu}^*(\mathbf{x}_t, t, \hat{\mathbf{y}}) - \boldsymbol{\nu}_\theta(\mathbf{x}_t, t, \hat{\mathbf{y}}) \big\|^2 \right] dt.
\end{align}
According to the Data Processing Inequality (DPI), we have:
\begin{equation}\label{eq:condupperbound}
D_{\text{KL}}(p(\mathbf{x}|\hat{\mathbf{y}}) \,\|\, q_\theta^s(\mathbf{x}|\hat{\mathbf{y}})) \le D_{\text{KL}}(\mathbb{P} \,\|\, \mathbb{Q})
\end{equation}
Based on the derivations in (\ref{eq:uncondKLwithcondKL}), we have
\begin{align}
  &D_{KL}(p(\mathbf{x})||q_\theta^s(\mathbf{x})) \leq \mathbb{E}_{p(\hat{\mathbf{y}})} \left[D_{\text{KL}}(p(\mathbf{x}|\hat{\mathbf{y}}) \,\|\, q_\theta^s(\mathbf{x}|\hat{\mathbf{y}})) \right]\nonumber\\
  &\leq\mathbb{E}_{p(\hat{\mathbf{y}})} \left[\int_0^T \frac{2}{g(t)^2 T^2} \mathbb{E}_{p_t(\mathbf{x}_t | \hat{\mathbf{y}})} \left[ \big\| \boldsymbol{\nu}^*(\cdot) - \boldsymbol{\nu}_\theta(\cdot) \big\|^2 \right] dt\right]\nonumber\\
  &= \int_0^T \frac{2}{g(t)^2 T^2} \mathbb{E}_{p_t(\mathbf{x}_t, \hat{\mathbf{y}})} \left[ \big\| \boldsymbol{\nu}^*(\mathbf{x}_t, t, \hat{\mathbf{y}}) - \boldsymbol{\nu}_\theta(\mathbf{x}_t, t, \hat{\mathbf{y}}) \big\|^2 \right] dt
\end{align}

\section{Proof of Proposition \ref{ODEKLbound}}\label{app:ODEKLbound}
For any fixed conditioning variable $\hat{\mathbf y}$, since the true ODE is
$\frac{d\mathbf{x}_t}{dt}=\frac{1}{T}\boldsymbol{\nu}^*(\mathbf{x}_t,t,\hat{\mathbf y})$,
the corresponding conditional marginal density $p_t(\mathbf{x}_t\mid \hat{\mathbf y})$ satisfies the continuity equation
\begin{equation}
\partial_t p_t(\mathbf{x}_t\mid \hat{\mathbf y})
+
\nabla_{\mathbf{x}_t}\cdot
\left(
p_t(\mathbf{x}_t\mid \hat{\mathbf y})
\frac{1}{T}\boldsymbol{\nu}^*(\mathbf{x}_t,t,\hat{\mathbf y})
\right)
=0.
\label{eq:true_continuity_residual}
\end{equation}
Likewise, since the learned ODE is $\frac{d\mathbf{x}_t}{dt}=\frac{1}{T}\boldsymbol{\nu}_\theta(\mathbf{x}_t,t,\hat{\mathbf y})$, the induced conditional marginal density $q_{t,\theta}^{o}(\mathbf{x}_t\mid \hat{\mathbf y})$ satisfies
\begin{equation}
\partial_t q_{t,\theta}^{o}(\mathbf{x}_t\mid \hat{\mathbf y})
+
\nabla_{\mathbf{x}_t}\cdot
\left(
q_{t,\theta}^{o}(\mathbf{x}_t\mid \hat{\mathbf y})
\frac{1}{T}\boldsymbol{\nu}_\theta(\mathbf{x}_t,t,\hat{\mathbf y})
\right)
=0.
\label{eq:learned_continuity_residual}
\end{equation}
Consider the gradient $\frac{d}{dt}D_{KL}\!\left(p_t(\mathbf{x}_t\mid \hat{\mathbf y})\|q_{t,\theta}^{o}(\mathbf{x}_t\mid \hat{\mathbf y})\right)$ is 
\begin{align}
&\int \partial_t p_t(\cdot)\log \frac{p_t(\cdot)}{q_{t,\theta}^{o}(\cdot)} +p_t(\cdot)\partial_t\left(\log \frac{p_t(\cdot)}{q_{t,\theta}^{o}(\cdot)}\right) d\mathbf{x}_t.
\label{eq:kl_derivative_step1_residual}
\end{align}
Since
\begin{align}
\partial_t\left(\log\frac{p_t(\mathbf{x}_t\mid \hat{\mathbf y})}{q_{t,\theta}^{o}(\mathbf{x}_t\mid \hat{\mathbf y})}\right)=\frac{\partial_t p_t(\mathbf{x}_t\mid \hat{\mathbf y})}{p_t(\mathbf{x}_t\mid \hat{\mathbf y})}-\frac{\partial_t q_{t,\theta}^{o}(\mathbf{x}_t\mid \hat{\mathbf y})}{q_{t,\theta}^{o}(\mathbf{x}_t\mid \hat{\mathbf y})},
\end{align}
the second term can be expanded as
\[\int p_t(\cdot)\partial_t\left(\log\frac{p_t(\cdot)}{q_{t,\theta}^{o}(\cdot)}\right)d\mathbf{x}_t=\int \partial_t p_t(\cdot) - p_t(\cdot)\frac{\partial_t q_{t,\theta}^{o}(\cdot)}{q_{t,\theta}^{o}(\cdot)}\,d\mathbf{x}_t.\]
As $p_t(\cdot)$ is a probability density, we have 
\begin{align}
\int_{\mathbb R^d}\partial_t p_t(\mathbf{x}_t\mid \hat{\mathbf y})\,d\mathbf{x}_t=\frac{d}{dt}\int_{\mathbb R^d}p_t(\mathbf{x}_t\mid \hat{\mathbf y})\,d\mathbf{x}_t=0.
\end{align}
Hence, the gradient $\frac{d}{dt}D_{KL}\!\left(p_t(\mathbf{x}_t\mid \hat{\mathbf y})\|q_{t,\theta}^{o}(\mathbf{x}_t\mid \hat{\mathbf y})\right)$ is 
\begin{align}
\int\partial_t p_t(\cdot)\log\frac{p_t(\cdot)}{q_{t,\theta}^{o}(\cdot)}\,d\mathbf{x}_t-\int p_t(\cdot)\frac{\partial_t q_{t,\theta}^{o}(\cdot)}{q_{t,\theta}^{o}(\cdot)}\,d\mathbf{x}_t.
\label{eq:kl_derivative_step3_residual}
\end{align}
Substituting \eqref{eq:true_continuity_residual} and \eqref{eq:learned_continuity_residual} into \eqref{eq:kl_derivative_step3_residual} yields
\begin{align}
&\frac{d}{dt}D_{KL}\!\left(p_t(\mathbf{x}_t\mid \hat{\mathbf y})\middle\|q_{t,\theta}^{o}(\mathbf{x}_t\mid \hat{\mathbf y})\right)\nonumber\\
&=-\int \nabla_{\mathbf{x}_t}\cdot\left(p_t(\mathbf{x}_t\mid \hat{\mathbf y})\frac{1}{T}\boldsymbol{\nu}^*(\mathbf{x}_t,t,\hat{\mathbf y})\right)\log\frac{p_t(\mathbf{x}_t\mid \hat{\mathbf y})}{q_{t,\theta}^{o}(\mathbf{x}_t\mid \hat{\mathbf y})}\,d\mathbf{x}_t\nonumber\\
&+\int p_t(\mathbf{x}_t\mid \hat{\mathbf y})\frac{\nabla_{\mathbf{x}_t}\cdot\left(q_{t,\theta}^{o}(\mathbf{x}_t\mid \hat{\mathbf y})\frac{1}{T}\boldsymbol{\nu}_\theta(\mathbf{x}_t,t,\hat{\mathbf y})\right)}{q_{t,\theta}^{o}(\mathbf{x}_t\mid \hat{\mathbf y})}\,d\mathbf{x}_t.
\label{eq:kl_derivative_step4_residual}
\end{align}
Through integration by parts, the first term in \eqref{eq:kl_derivative_step4_residual} is
\begin{align}
\int_{\mathbb R^d}p_t(\mathbf{x}_t\mid \hat{\mathbf y})\left(\frac{1}{T}\boldsymbol{\nu}^*(\mathbf{x}_t,t,\hat{\mathbf y})\right)^\top\nabla_{\mathbf{x}_t}\log\frac{p_t(\mathbf{x}_t\mid \hat{\mathbf y})}{q_{t,\theta}^{o}(\mathbf{x}_t\mid \hat{\mathbf y})}\,d\mathbf{x}_t.
\label{eq:first_term_parts_residual}
\end{align}
For the second term, we have
\begin{align}
&\frac{\nabla_{\mathbf{x}_t}\cdot\left(q_{t,\theta}^{o}(\mathbf{x}_t\mid \hat{\mathbf y})\frac{1}{T}\boldsymbol{\nu}_\theta(\mathbf{x}_t,t,\hat{\mathbf y})\right)}{q_{t,\theta}^{o}(\mathbf{x}_t\mid \hat{\mathbf y})}=\nabla_{\mathbf{x}_t}\cdot\left(\frac{1}{T}\boldsymbol{\nu}_\theta(\mathbf{x}_t,t,\hat{\mathbf y})\right)\nonumber\\
&+\left(
\frac{1}{T}\boldsymbol{\nu}_\theta(\mathbf{x}_t,t,\hat{\mathbf y})\right)^\top\nabla_{\mathbf{x}_t}\log q_{t,\theta}^{o}(\mathbf{x}_t\mid \hat{\mathbf y}).
\end{align}
Therefore, the second term in \eqref{eq:kl_derivative_step4_residual} can be expressed as
\begin{align}
&\int_{\mathbb R^d}p_t(\mathbf{x}_t\mid \hat{\mathbf y})\frac{\nabla_{\mathbf{x}_t}\cdot\left(q_{t,\theta}^{o}(\mathbf{x}_t\mid \hat{\mathbf y})\frac{1}{T}\boldsymbol{\nu}_\theta(\mathbf{x}_t,t,\hat{\mathbf y})\right)}{q_{t,\theta}^{o}(\mathbf{x}_t\mid \hat{\mathbf y})}\,d\mathbf{x}_t\nonumber\\
&=\int p_t(\mathbf{x}_t\mid \hat{\mathbf y})\nabla_{\mathbf{x}_t}\cdot\left(\frac{1}{T}\boldsymbol{\nu}_\theta(\mathbf{x}_t,t,\hat{\mathbf y})\right)\,d\mathbf{x}_t\nonumber\\
&+\int p_t(\mathbf{x}_t\mid \hat{\mathbf y})\left(\frac{1}{T}\boldsymbol{\nu}_\theta(\mathbf{x}_t,t,\hat{\mathbf y})\right)^\top\nabla_{\mathbf{x}_t}\log q_{t,\theta}^{o}(\mathbf{x}_t\mid \hat{\mathbf y})\,d\mathbf{x}_t.\nonumber\\
&=-\int p_t(\mathbf{x}_t\mid \hat{\mathbf y})\left(\frac{1}{T}\boldsymbol{\nu}_\theta(\mathbf{x}_t,t,\hat{\mathbf y})\right)^\top\nabla_{\mathbf{x}_t}\log\frac{p_t(\mathbf{x}_t\mid \hat{\mathbf y})}{q_{t,\theta}^{o}(\mathbf{x}_t\mid \hat{\mathbf y})}\,d\mathbf{x}_t.
\label{eq:second_term_split_residual}
\end{align}
Combining \eqref{eq:first_term_parts_residual} and \eqref{eq:second_term_split_residual} with \eqref{eq:kl_derivative_step4_residual} gives
\begin{align}
&\frac{d}{dt}D_{KL}\!\left(p_t(\mathbf{x}_t\mid \hat{\mathbf y})\middle\|q_{t,\theta}^{o}(\mathbf{x}_t\mid \hat{\mathbf y})\right)=\frac{1}{T}\mathbb E_{p_t(\mathbf{x}_t\mid \hat{\mathbf y})}\cdot\nonumber\\
&\left[\left(\boldsymbol{\nu}^*(\mathbf{x}_t,t,\hat{\mathbf y})-\boldsymbol{\nu}_\theta(\mathbf{x}_t,t,\hat{\mathbf y})\right)^\top\nabla_{\mathbf{x}_t}\log\frac{p_t(\mathbf{x}_t\mid \hat{\mathbf y})}{q_{t,\theta}^{o}(\mathbf{x}_t\mid \hat{\mathbf y})}\right].
\label{eq:kl_derivative_exact_residual}
\end{align}
With $\nabla_{\mathbf{x}_t}\log p_t(\mathbf{x}_t\mid \hat{\mathbf y})=-\frac{T-t}{t}\boldsymbol{\nu}^*(\mathbf{x}_t,t,\hat{\mathbf y})-\frac{T}{t}\mathbf{x}_t$,
we have
\[\nabla_{\mathbf{x}_t}\log\frac{p_t(\cdot)}{q_{t,\theta}^{o}(\cdot)}=-\frac{T-t}{t}\boldsymbol{\nu}^*(\cdot)-\frac{T}{t}\mathbf{x}_t-\nabla_{\mathbf{x}_t}\log q_{t,\theta}^{o}(\cdot).\]
Now, $\frac{d}{dt}D_{KL}\!\left(p_t(\mathbf{x}_t\mid \hat{\mathbf y})\|q_{t,\theta}^{o}(\mathbf{x}_t\mid \hat{\mathbf y})\right)$ can be expressed as 

\begin{align}
&\frac{1}{T}\mathbb E_{p_t(\cdot)}[\left(\boldsymbol{\nu}^*(\cdot)-\boldsymbol{\nu}_\theta(\cdot)\right)^\top (-\frac{T-t}{t}\left(\boldsymbol{\nu}^*(\cdot)-\boldsymbol{\nu}_\theta(\cdot)\right)-\mathbf r_\theta(\cdot))].
\end{align}
Expanding the inner product yields and integrating both sides with respect to $t$ from $0$ to $T$ gives
\begin{align}
&D_{KL}\left(p_0(\mathbf{x}_0\mid \hat{\mathbf y})\middle\|q_{0,\theta}^{o}(\mathbf{x}_0\mid \hat{\mathbf y})\right)\nonumber\\
&=D_{KL}\left(p_T(\mathbf{x}_T\mid \hat{\mathbf y})\middle\|q_{T,\theta}^{o}(\mathbf{x}_T\mid \hat{\mathbf y})\right)\nonumber\\
&+\int_{0}^{T}\frac{T-t}{Tt}\mathbb{E}_{p_t(\mathbf{x}_t, \hat{\mathbf y})}\left[\left\|\boldsymbol{\nu}^*(\mathbf{x}_t,t,\hat{\mathbf y})-\boldsymbol{\nu}_\theta(\mathbf{x}_t,t,\hat{\mathbf y})\right\|^2\right]\nonumber\\
&+\frac{1}{T}\mathbb{E}_{p_t(\mathbf{x}_t, \hat{\mathbf y})}\left[\left(\boldsymbol{\nu}^*(\mathbf{x}_t,t,\hat{\mathbf y})-\boldsymbol{\nu}_\theta(\mathbf{x}_t,t,\hat{\mathbf y})\right)^\top\mathbf r_\theta(\mathbf{x}_t,t,\hat{\mathbf y})\right]dt .
\end{align}
Since $\mathbf{x}_T$ is sampled from the same distribution $\pi(\mathbf{x}_T)$ in both ODEs, the term $D_{KL}\left(p_T(\mathbf{x}_T\mid \hat{\mathbf y})\middle\|q_{T,\theta}^{o}(\mathbf{x}_T\mid \hat{\mathbf y})\right)$ is 0.

Finally, because the condition $\mathbf{\hat{y}}$ is the received signal in our scenario, similar to the derivation in \eqref{eq:uncondKLwithcondKL}, we have
\begin{align}
  D_{KL}\left(p(\mathbf{x})\middle\|q_{\theta}^{o}(\mathbf{x})\right)\leq\mathbb{E}_{p(\hat{\mathbf{y}})}D_{KL}\left(p_0(\mathbf{x}_0\mid \hat{\mathbf y})\middle\|q_{0,\theta}^{o}(\mathbf{x}_0\mid \hat{\mathbf y})\right).\nonumber
\end{align}
Therefore, 
\begin{align}
&D_{KL}\left(p(\mathbf{x})\middle\|q_{\theta}^{o}(\mathbf{x})\right)\nonumber\\
&\leq\int_{0}^{T}\frac{T-t}{Tt}\mathbb{E}_{p_t(\mathbf{x}_t, \hat{\mathbf y})}\left[\left\|\boldsymbol{\nu}^*(\mathbf{x}_t,t,\hat{\mathbf y})-\boldsymbol{\nu}_\theta(\mathbf{x}_t,t,\hat{\mathbf y})\right\|^2\right]\nonumber\\
&+\frac{1}{T}\mathbb{E}_{p_t(\mathbf{x}_t, \hat{\mathbf y})}\left[\left(\boldsymbol{\nu}^*(\mathbf{x}_t,t,\hat{\mathbf y})-\boldsymbol{\nu}_\theta(\mathbf{x}_t,t,\hat{\mathbf y})\right)^\top\mathbf r_\theta(\mathbf{x}_t,t,\hat{\mathbf y})\right]\,dt .\nonumber
\end{align}

\section{Proof of Proposition \ref{ODEW2bound}}\label{app:ODEW2bound}
We construct a coupling by fixing the same initial point, thereby obtaining an upper bound on the $W_2$ distance.
For any fixed $\hat{\mathbf{y}}$, let the two ODEs share the same initial point $\mathbf{x}_T$. Therefore, the two ODEs are
\begin{align}
&\frac{d\mathbf{x}_t}{dt}=\frac{1}{T}\boldsymbol{\nu}^*(\mathbf{x}_t,t,\hat{\mathbf{y}}),\qquad\mathbf{x}_T\sim \pi(\mathbf{x}_T),\\
&\frac{d\mathbf{z}_t}{dt}=\frac{1}{T}\boldsymbol{\nu}_{\theta}(\mathbf{z}_t,t,\hat{\mathbf{y}}),\qquad\mathbf{z}_T=\mathbf{x}_T.
\end{align}
By construction, the joint law of $(\mathbf{X}_t,\mathbf{Z}_t)$ is a coupling of $p_t(\mathbf{x}_t\mid \hat{\mathbf{y}})$ and $q_{t,\theta}^{o}(\mathbf{z}_t\mid \hat{\mathbf{y}})$. By the definition of the $2$-Wasserstein distance,
\begin{align}
W_2^2\!\left(p_t(\mathbf{x}_t\mid \hat{\mathbf{y}}),q_{t,\theta}^{o}(\mathbf{z}_t\mid \hat{\mathbf{y}})\right)\le\mathbb{E}\!\left[\left\|\mathbf{x}_t-\mathbf{z}_t\right\|^2\right].
\label{eq:w2_coupling_upper_bound}
\end{align}
Now define $\boldsymbol{\Delta}_t:=\mathbf{x}_t-\mathbf{z}_t$.
Then
\begin{align}
\frac{d\boldsymbol{\Delta}_t}{dt}&=\frac{1}{T}\left(\boldsymbol{\nu}^*(\mathbf{x}_t,t,\hat{\mathbf y})-\boldsymbol{\nu}_{\theta}(\mathbf{z}_t,t,\hat{\mathbf y})\right).
\label{eq:delta_dynamics}
\end{align}
Hence,
by inserting and subtracting $\boldsymbol{\nu}_{\theta}(\mathbf{x}_t,t,\hat{\mathbf y})$, we have
\begin{align}
-\frac{d}{dt}\left\|\boldsymbol{\Delta}_t\right\|^2&=\frac{2}{T}\boldsymbol{\Delta}_t^\top\left(\boldsymbol{\nu}_{\theta}(\mathbf{x}_t,t,\hat{\mathbf y})-\boldsymbol{\nu}^*(\mathbf{x}_t,t,\hat{\mathbf y})\right)\nonumber\\
&\quad+\frac{2}{T}\boldsymbol{\Delta}_t^\top\left(\boldsymbol{\nu}_{\theta}(\mathbf{z}_t,t,\hat{\mathbf y})-\boldsymbol{\nu}_{\theta}(\mathbf{x}_t,t,\hat{\mathbf y})\right).
\label{eq:delta_norm_derivative_step2}
\end{align}
We bound the first term in \eqref{eq:delta_norm_derivative_step2} by Young's inequality, yielding
\begin{align}
&\frac{2}{T}\boldsymbol{\Delta}_t^\top\left(\boldsymbol{\nu}_{\theta}(\mathbf{x}_t,t,\hat{\mathbf y})-\boldsymbol{\nu}^*(\mathbf{x}_t,t,\hat{\mathbf y})\right)\nonumber\\
&\le\frac{1}{T}\left\|\boldsymbol{\Delta}_t\right\|^2+\frac{1}{T}
\left\|\boldsymbol{\nu}_{\theta}(\mathbf{x}_t,t,\hat{\mathbf y})-\boldsymbol{\nu}^*(\mathbf{x}_t,t,\hat{\mathbf y})\right\|^2.
\label{eq:first_term_young_bound}
\end{align}

For the second term in \eqref{eq:delta_norm_derivative_step2}, by the Lipschitz continuity assumption of $\boldsymbol{\nu}_{\theta}(\mathbf{x}_t,t,\hat{\mathbf y})$ with respect to $\mathbf{x}_t$, we have
\begin{align}
&\left\|\boldsymbol{\nu}_{\theta}(\mathbf{z}_t,t,\hat{\mathbf y})-\boldsymbol{\nu}_{\theta}(\mathbf{x}_t,t,\hat{\mathbf y})\right\|\le L_p(t)\left\|\mathbf{z}_t-\mathbf{x}_t\right\| \label{eq:lipschitz_theta_field}
\end{align}
Therefore,
\begin{align}
&\boldsymbol{\Delta}_t^\top\left(\boldsymbol{\nu}_{\theta}(\mathbf{z}_t,t,\hat{\mathbf{y}})-\boldsymbol{\nu}_{\theta}(\mathbf{x}_t,t,\hat{\mathbf{y}})\right)\nonumber\\
&\le\left\|\boldsymbol{\Delta}_t\right\|\left\|\boldsymbol{\nu}_{\theta}(\mathbf{z}_t,t,\hat{\mathbf{y}})-\boldsymbol{\nu}_{\theta}(\mathbf{x}_t,t,\hat{\mathbf{y}})\right\| \le L_p(t)\left\|\boldsymbol{\Delta}_t\right\|^2.
\label{eq:second_term_lipschitz_bound}
\end{align}

Combine \eqref{eq:delta_norm_derivative_step2}, \eqref{eq:first_term_young_bound}, and \eqref{eq:second_term_lipschitz_bound} and take expectations, yielding
\begin{align}
-\frac{d}{dt}\mathbb{E}\!\left[\left\|\mathbf{x}_t-\mathbf{z}_t\right\|^2\right]\le\frac{2L_p(t)+1}{T}\mathbb{E}\!\left[\left\|\mathbf{x}_t-\mathbf{z}_t\right\|^2\right]&\nonumber\\
+\frac{1}{T}\mathbb{E}_{p_t(\mathbf{x}_t\mid \hat{\mathbf y})}\left[\left\|\boldsymbol{\nu}^*(\mathbf{x}_t,t,\hat{\mathbf y})-\boldsymbol{\nu}_{\theta}(\mathbf{x}_t,t,\hat{\mathbf y})\right\|^2\right]&.
\label{eq:expected_delta_inequality_step2}
\end{align}

The function $\mathbb{E}[\left\|\mathbf{x}_t-\mathbf{z}_t\right\|^2]$ matches the backward Gronwall inequality and the initial term $\mathbb{E}[\left\|\mathbf{x}_T-\mathbf{z}_T\right\|^2]=0$.
We have
\begin{align}
\mathbb{E}\!\left[\left\|\mathbf{x}_t-\mathbf{z}_t\right\|^2\right]\le\int_t^T\exp\!\left(\int_t^s \frac{2L_p(\tau)+1}{T}\,d\tau\right)&\nonumber\\
\cdot\frac{1}{T}\mathbb{E}_{p_s(\mathbf{x}_s\mid \hat{\mathbf y})}\left[\left\|\boldsymbol{\nu}^*(\mathbf{x}_s,s,\hat{\mathbf y})-\boldsymbol{\nu}_{\theta}(\mathbf{x}_s,s,\hat{\mathbf y})\right\|^2\right]ds.&
\label{eq:expected_delta_gronwall_result}
\end{align}

Combining \eqref{eq:w2_coupling_upper_bound} and \eqref{eq:expected_delta_gronwall_result}, we conclude that
\begin{align}
&W_2^2\!\left(p_t(\mathbf{x}_t\mid \hat{\mathbf y}),q_{t,\theta}^{o}(\mathbf{x}_t\mid \hat{\mathbf y})\right)\nonumber\\
&\le\int_t^T\exp\!\left(\int_t^s \frac{2L_p(\tau)+1}{T}\,d\tau\right)\nonumber\\
&\qquad\cdot\frac{1}{T}\mathbb{E}_{p_s(\mathbf{x}_s\mid \hat{\mathbf y})}\left[
\left\|\boldsymbol{\nu}^*(\mathbf{x}_s,s,\hat{\mathbf y})-\boldsymbol{\nu}_{\theta}(\mathbf{x}_s,s,\hat{\mathbf y})\right\|^2\right]ds.
\end{align}

Finally, in our scenario, the condition $\hat{\mathbf{y}}$ is the received signal shared by both distributions and the function is \(L_p(t)\)-Lipschitz uniformly over $\mathbf{\hat{y}}$. Therefore, we have
\begin{align}
&W_2^2\!\left(p(\mathbf{x}),q_{\theta}^{o}(\mathbf{x})\right)
\leq\mathbb{E}_{p(\hat{\mathbf{y}})}\left[W_2^2\!\left(p_0(\mathbf{x}_0\mid \hat{\mathbf y}),q_{0,\theta}^{o}(\mathbf{x}_0\mid
\hat{\mathbf y})\right)\right]\nonumber\\
&\le\int_0^T\exp\!\left(\int_0^s \frac{2L_p(\tau)+1}{T}\,d\tau\right)\nonumber\\&
\cdot\frac{1}{T}\mathbb{E}_{p_s(\mathbf{x}_s, \hat{\mathbf y})}\left[\left\|\boldsymbol{\nu}^*(\mathbf{x}_s,s,\hat{\mathbf y})-\boldsymbol{\nu}_{\theta}(\mathbf{x}_s,s,\hat{\mathbf y})\right\|^2\right]ds.
\end{align}

\footnotesize
\bibliographystyle{IEEEtran}
\bibliography{reference}{}
\end{document}